\documentclass[11pt]{article}

\usepackage[utf8]{inputenc}
\usepackage[T1]{fontenc}
\usepackage{amsmath,amssymb,amsthm}
\usepackage{mathtools}
\usepackage{bm}
\usepackage{graphicx}
\usepackage[margin=1in]{geometry}
\usepackage[colorlinks=true,linkcolor=blue,citecolor=blue,urlcolor=blue]{hyperref}
\usepackage{booktabs}
\usepackage{caption}
\usepackage{float}
\usepackage{enumitem}
\usepackage{xcolor}
\usepackage{physics}

\newtheorem{theorem}{Theorem}
\newtheorem{lemma}{Lemma}
\newtheorem{corollary}{Corollary}
\newtheorem{proposition}{Proposition}
\newtheorem{definition}{Definition}
\newtheorem{hypothesis}{Hypothesis}
\theoremstyle{remark}
\newtheorem{remark}{Remark}

\newcommand{\rhat}{\hat{r}}
\newcommand{\ahat}{\hat{a}}
\newcommand{\Fclass}[2]{\mathcal{F}^{(#1)}_{#2}}
\newcommand{\Kerr}{\mathrm{Kerr}}
\newcommand{\ordc}{\mathrm{ord}_c}
\newcommand{\Var}{\mathrm{Var}}
\newcommand{\Sym}{\mathrm{Sym}}
\newcommand{\nustar}{\nu^{\star}}
\newcommand{\Cz}{\mathcal{C}_0}

\title{\textbf{Computational Superiority of Non-Markovian Kerr Feedback\\
in Continuous-Variable Quantum Reservoir Computing}}
\author{Daniel Soh\\
\small Wyant College of Optical Sciences, University of Arizona}
\date{May 26, 2026}

\begin{document}
\maketitle

\begin{abstract}
A linear optical medium can delay, mix, and superpose light, but it can never make
two pulses of light multiply one another: multiplication is nonlinear, and a linear
system has no such operation. This is the physical root of a sharp limitation on
continuous-variable quantum reservoir computers (QRCs) built from Gaussian
optics---they cannot, within the reservoir, form genuine products of the input at
different past times, the cross-time nonlinear correlations that many temporal
computations require. They can only fake such a product by storing each past input
separately and multiplying in the readout, which forces an exponentially harder
high-order measurement. We show that adding a single Kerr (intensity-dependent
phase) element in a time-delayed feedback loop removes this limitation, and does so
with a striking economy. Because the Kerr effect makes light's phase depend on its
own intensity, it performs a true multiplication inside the medium; and because the
feedback makes the light revisit that element many times, one physical mode mixes
its own history against itself once per round-trip. Feedback thereby turns time into
space: $D$ passes through one nonlinear mode substitute for $D$ parallel linear
modes. We prove this as an unbounded resource separation (Theorem~\ref{thm:scaling},
Corollary~\ref{cor:unbounded}): a Gaussian reservoir of $N$ modes can reach
cross-time nonlinear structure of ``rank'' at most $2N$---a ceiling set purely by
its hardware---while a single Kerr mode reaches rank equal to its feedback depth
$D$, a free parameter that costs no extra modes. For every $N$ there is a
computation one Kerr mode performs that no $N$-mode linear reservoir can; no finite
linear reservoir suffices. A counterintuitive ingredient makes the mechanism work:
loss. Each round-trip dims the light slightly, so the nonlinear phase it accumulates
differs from pass to pass, giving every echo its own fingerprint; without loss the
passes would be identical and redundant. We confirm the activation of the effect on
an exact open-system simulation, identify the single-photon Kerr-to-loss ratio as
the governing figure of merit, and ground the separation in a recognized task,
nonlinear channel equalization, where the required computational rank is exactly a
channel's nonlinear memory complexity. The benefit is unbounded in principle and
capped in practice by physics: achievable feedback depth is $D\sim30$--$230$ on
integrated platforms, so one nonlinear mode replaces up to $\sim10^2$ linear ones,
at the price of additional measurement time---the faint per-mode imprint must be
read out over many shots. The architecture thus trades a hardware resource (modes,
detectors, chip area) for a measurement resource (shots), a favorable trade exactly
where mode count is the binding constraint. We are explicit throughout about the
boundary between the unbounded separation we prove and the physically-capped reach
of any real device.
\end{abstract}

\tableofcontents
\clearpage

\section{Introduction}
\label{sec:intro}

Reservoir computing processes time-varying signals by letting them drive a fixed
physical system---the reservoir---and training only a simple readout of the
system's state. The appeal is hardware: the reservoir itself need not be tuned, so
a well-chosen physical medium does the hard part of the computation for free.
Continuous-variable quantum optics is an attractive medium, and a sequence of
results has established that even a purely \emph{Gaussian} optical reservoir---beam
splitters, squeezers, and linear loss---is a universal approximator of fading-memory
maps when the input is encoded suitably~\cite{nokkala2021}. Universality, however,
says nothing about \emph{cost}: it guarantees that a sufficiently large reservoir
can approximate a target, not that a small one can. The question this paper answers
is what a \emph{small} reservoir can do, and what a single nonlinear element buys.

\paragraph{Why a linear reservoir struggles with cross-time products.}
Many temporal computations require the system to form a \emph{product of the input
at two different past times}---to compute, say, $u_{k-2}\,u_{k-4}$, or any
cross-time nonlinear correlation. A Gaussian reservoir cannot do this within the
reservoir, and the reason is physical, not technical: a Gaussian system evolves its
light by linear maps and adds noise linearly, so the only correlation it can create
between two times is the one a linear filter makes---a second-order covariance. It
has no operation that multiplies the field at time $A$ by the field at time $B$,
because multiplication is nonlinear and the medium is linear. The only recourse is
to keep both past inputs as separate stored features and multiply them in the
\emph{readout}. But a product of $\nu$ stored quantities is a degree-$\nu$ readout
feature, and reading a degree-$\nu$ polynomial of a quadrature to fixed precision is
exponentially expensive in $\nu$. So a linear reservoir does not lack the
computation outright---it pays for it in measurement, at the back end, where it
hurts most.

\paragraph{Why one Kerr mode in a loop escapes this---and acts like many modes.}
A Kerr element makes light's phase depend on its own intensity; this is a genuine
multiplication of the field by itself, performed \emph{inside} the medium rather
than deferred to the readout. Placed in a time-delayed feedback loop, the light
circulates and passes through the Kerr element repeatedly, and on each pass the mode
mixes its present state against an echo of its earlier state. After $D$ round-trips,
a single physical mode has performed $D$ distinct nonlinear mixings of its own
history. A linear reservoir would need $D$ separate modes to hold $D$ independent
channels of cross-time structure; the feedback loop manufactures them from one mode
reused $D$ times. \emph{Feedback turns time into space.} This is the source of the
resource separation we prove, and of its unbounded character: looping more times
adds processing channels without adding hardware, so for any number of modes a
linear reservoir might have, a Kerr loop with enough depth surpasses it.

\paragraph{Why loss is the friend, not the enemy.}
The mechanism would fail in one circumstance: if every round-trip were identical, the
$D$ passes would imprint the \emph{same} nonlinear phase and collapse to a single
redundant channel. What makes the passes distinct is that the circulating light dims
a little each round-trip, so the intensity-dependent Kerr phase it accumulates is
slightly different on every pass. Loss---ordinarily the adversary of quantum
coherence---is here precisely what gives each echo its own fingerprint and makes the
$D$ channels independent. The size of the advantage is therefore governed by a
single physical ratio, the per-photon Kerr phase relative to the loss per
round-trip, which sets both how strongly the effect activates and how many distinct
round-trips survive.

\paragraph{The catch: a faint imprint costs measurement time.}
The cross-time structure a Kerr loop creates is genuinely present in the output
light, but a single mode carrying many channels imprints each one faintly. Reading
the channels apart takes many measurement repetitions. The architecture therefore
does not deliver a free lunch; it makes a \emph{trade}. It saves the resources that
are scarce in an integrated photonic chip---physical modes, detector chains, chip
area---and spends a resource that is comparatively cheap there---measurement time.
Where mode count is the binding constraint, as it is in integrated photonics, this
is a favorable exchange. We quantify both sides of it.

\paragraph{Relation to prior work.}
Quantum reservoir computing was introduced by Fujii and Nakajima~\cite{fujii2017}
using the density-matrix space of a disordered qubit ensemble, and was developed
in many directions: the reservoir-processing framework of Ghosh
et~al.~\cite{ghosh2019}, the role of thermalization and dynamical phase
transitions~\cite{martinezpena2021}, weak- and projective-measurement readout
strategies~\cite{mujal2023}, and reservoir approaches to quantum state
measurement~\cite{angelatos2021}. Our work sits in the continuous-variable branch,
whose foundational result---that a purely Gaussian CV reservoir is already a
universal fading-memory approximator under suitable input
encoding~\cite{nokkala2021}---is precisely the premise that motivates our question:
since universality is settled, what is gained is a matter of resources, not of
reachability in the limit.

Closest to the present setting is the line of work placing a Kerr nonlinearity
\emph{in the reservoir itself}. Govia et~al.~\cite{govia2021} use a single Kerr
oscillator as the reservoir; coupled-Kerr-oscillator reservoirs have been studied
for the role of entanglement and dissipation~\cite{karimi2025}. These works locate
the nonlinearity in the reservoir's own Hamiltonian, typically Markovian and often
in a semiclassical, many-photon regime. Our architecture differs in a specific and
consequential way: the reservoir is held \emph{Gaussian}, and a single Kerr element
is placed in a \emph{time-delayed, non-Markovian feedback arm}. This makes the
placement of the nonlinearity an explicit design variable and isolates its
computational role as the sole difference from a provably weaker Gaussian baseline,
rather than entangling it with the reservoir's intrinsic dynamics. It is this
placement that yields the unbounded separation we prove; the comparison is
feedback-on versus feedback-off within one architecture.

The present work also builds directly on a line of reservoir-computing studies from
our group in which feedback is the central design element. State feedback was shown
to enhance classical echo-state networks~\cite{ehlers2025statefeedback}, and a
general theory of stochastic reservoir computers was developed
in~\cite{ehlers2025stochastic}; on the quantum side, we established practical
few-atom quantum reservoir computing~\cite{zhu2025fewatom} and, most directly
relevant here, showed that coherent feedback enhances minimalistic and scalable
quantum reservoir computing~\cite{zhu2025feedback}. All-optical echo-state
realizations~\cite{kaushik2026alloptical} and all-optical nonlinearities from
atom--cavity interactions~\cite{zhu2025qonn} supply the photonic-platform context.
The present paper sharpens this program into an exact statement: it identifies
\emph{where} the feedback nonlinearity must be placed to gain expressive power that
no linear reservoir of any size possesses, and proves the resulting separation is
unbounded.

We also distinguish our claim from the well-established classical \emph{time-delay
reservoir computing} paradigm, in which a single nonlinear node with delayed
feedback time-multiplexes ``virtual nodes'' to emulate a many-node
reservoir~\cite{appeltant2011,larger2012}. That body of work is an \emph{equivalence
and efficiency} result---one physical node can realize the function of a spatially
extended reservoir---and is the conceptual ancestor of our ``feedback turns time
into space'' picture. Our contribution is of a different kind: a proven
\emph{separation}. We show not that a feedback reservoir matches a comparable one,
but that a single non-Markovian Kerr mode reaches a class of cross-time nonlinear
computations that \emph{no} finite-mode Gaussian reservoir reaches at fixed readout
degree---an expressivity statement, quantified in the field-standard currency of
information-processing capacity~\cite{dambre2012} and reachable-kernel rank, rather
than a benchmarking equivalence.

We make the picture above precise. After fixing the model and its standing
hypotheses (Section~\ref{sec:setup}), we adopt the cumulant tower of the quadratures
as the reservoir state and expand the input--output map as a Volterra series
(Sections~\ref{sec:gradings}--\ref{sec:lemmas}); in this language the linear
reservoir's limitation is the statement that its connected cumulants above order two
vanish, and the Kerr feedback's power is that it does not. We prove the resulting
structural separation of reachable function classes (Section~\ref{sec:theorem}),
give a constructive account of what the enriched class can compute
(Section~\ref{sec:universality}), and specify an exact open-system model of the
device (Section~\ref{sec:model}) which we simulate to confirm the mechanism
(Section~\ref{sec:empirical}) and to locate the favorable hardware regime
(Section~\ref{sec:platform}). The central result (Section~\ref{sec:scaling}) is the
unbounded resource separation: feedback depth, not mode count, sets the reachable
cross-time rank, so one nonlinear mode outreaches any fixed linear reservoir. We
close by quantifying the physical caps on this benefit and the measurement price of
realizing it (Section~\ref{sec:ledger}), and by stating plainly where the proven
separation ends and operational reach begins (Section~\ref{sec:operational}).

\section{Setup and standing hypotheses}
\label{sec:setup}

\begin{figure}[H]
\centering
\includegraphics[width=0.92\linewidth]{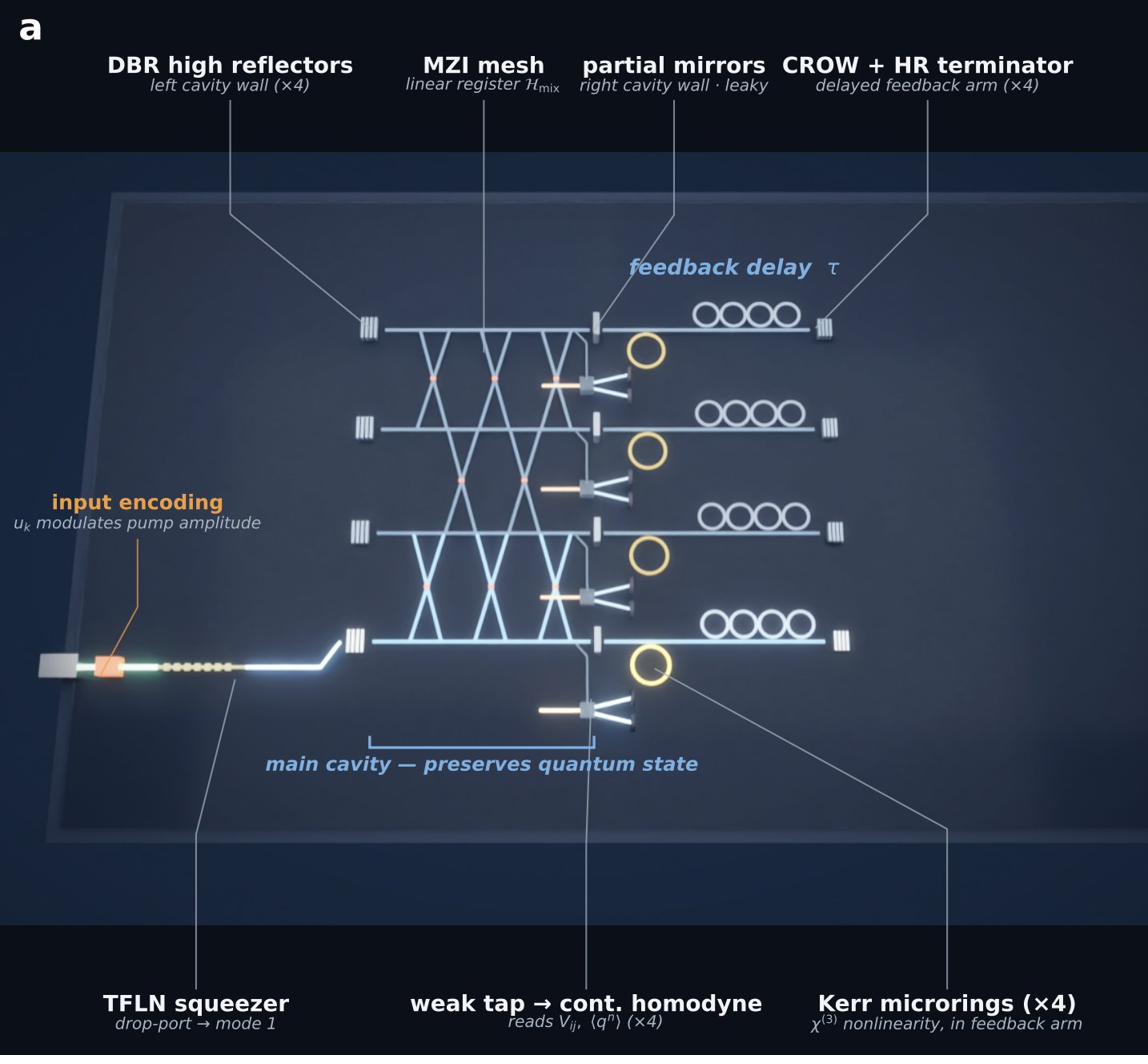}
\caption{Integrated-photonic realization of the non-Markovian Kerr-feedback
continuous-variable QRC. The classical input $u_k$ modulates the pump amplitude
of a thin-film lithium niobate (TFLN) parametric squeezer, injecting an
exponentially nonlinear single-mode Gaussian ancilla through the drop port into
mode~1 (Hypothesis~\ref{H1}). The fixed linear Gaussian register is realized by an
MZI mesh ($\mathcal{H}_{\mathrm{mix}}$) between DBR high reflectors forming the
main cavity, which preserves the quantum state (contraction set by intrinsic
loss, Hypothesis~\ref{H2}). A leaky right cavity wall (partial mirrors) couples
each mode both to a weak homodyne tap---continuously reading the quadrature
moments $\langle q^n\rangle$ that form the degree-$d$ readout
features~\eqref{eq:readout}---and to the delayed coherent feedback arm: a
$\chi^{(3)}$ Kerr microring followed by a CROW delay line of round-trip delay
$\tau$ terminated by a high reflector. The Kerr nonlinearity in the feedback arm,
through coupling strength $\chi$, is the single knob distinguishing the Gaussian
baseline ($\chi=0$) from the NM-Kerr reservoir.}
\label{fig:system}
\end{figure}

\paragraph{Reservoir and quadratures.}
The reservoir consists of $N$ bosonic modes with annihilation operators
$\ahat_1,\dots,\ahat_N$ and canonical quadratures collected in the vector
\begin{equation}
\rhat = (q_1,p_1,\dots,q_N,p_N)^\top,\qquad
q_j = \tfrac{1}{\sqrt2}(\ahat_j+\ahat_j^\dagger),\quad
p_j = \tfrac{-i}{\sqrt2}(\ahat_j-\ahat_j^\dagger),
\end{equation}
with $[\rhat_\mu,\rhat_\nu]=i\Omega_{\mu\nu}$ and $\Omega$ the standard symplectic
form ($\hbar=1$).

\paragraph{State variable: the cumulant tower.}
Let $\Sym^n(\mathbb{R}^{2N})$ denote the symmetric $n$-th tensor power. The state
variable is the tower of Weyl-ordered (symmetric) cumulants of the quadratures,
\begin{equation}
\kappa = \big(\kappa^{(1)},\kappa^{(2)},\kappa^{(3)},\dots\big),
\qquad \kappa^{(n)}\in\Sym^n(\mathbb{R}^{2N}),
\label{eq:tower}
\end{equation}
where $\kappa^{(1)}=\langle\rhat\rangle$ is the mean vector, $\kappa^{(2)}$ the
covariance matrix, and $\kappa^{(n\ge3)}$ the connected $n$-point functions. The
moment tower and the cumulant tower are related by a fixed, polynomial,
grading-preserving bijection (the moment--cumulant formula), so \eqref{eq:tower}
is an equivalent encoding of all quadrature moments. Throughout, the integer
$n=|\alpha|$ grading $\kappa^{(n)}$ is the \emph{operator order} (degree of the
underlying $q,p$ polynomial); it is the quantity limited by the homodyne readout
degree $d$.

\paragraph{Input and sequential drive.}
A scalar input sequence $\{u_k\}_{k\in\mathbb{Z}}\subset\mathbb{K}$ (a compact
input set) is fed one symbol per step of duration $T=T_{\mathrm{on}}+T_{\mathrm{off}}$.
At step $k$ the encoder injects, through the drop port into mode~1, a single-mode
Gaussian ancilla with displacement $\beta(u_k)$ and squeezing $(r(u_k),\phi)$,
$r(u)=g\,u$. The injected first- and second-order data are smooth maps
\begin{equation}
\beta:\mathbb{R}\to\mathbb{R}^{2N},\qquad
\gamma:\mathbb{R}\to\Sym^2(\mathbb{R}^{2N}),
\label{eq:inject}
\end{equation}
supported on the mode-1 block. The covariance injection $\gamma(u)$ contains the
squeezing factors $\cosh 2r(u),\ \sinh 2r(u)$ and is therefore an entire
(exponentially nonlinear) function of $u$. We write the Taylor data
$\beta_1:=\beta'(0)$, $\gamma_1:=\gamma'(0)$.

\paragraph{Feedback arm.}
A leaky port couples each mode to a feedback waveguide containing a Kerr element
($H_K=\tfrac{\chi}{2}\sum_j \ahat_j^{\dagger2}\ahat_j^2$, so
$\dot{\ahat}_j=-i\chi\,\ahat_j^\dagger \ahat_j\ahat_j+\cdots$) and a delay line of
round-trip delay $\tau$, terminated by a high reflector. We write
$\ell_\tau:=\tau/T\in\mathbb{N}$ for the delay in input steps. The parameter
$\chi$ is the single knob that distinguishes the two reservoirs compared below.

We make the following standing hypotheses.

\begin{hypothesis}[Injected encoding]\label{H1}
The input enters only through the injected ancilla data \eqref{eq:inject}.
Consequently the per-step reservoir propagation (mirror coupling, the MZI linear
network, intrinsic loss, and the ancilla reset) is a fixed, input-independent
Gaussian channel $\Cz$.
\end{hypothesis}

\begin{hypothesis}[Echo state / contraction]\label{H2}
The single-step covariance propagator $S_1$ induced by $\Cz$ on
$\mathbb{R}^{2N}$ has spectral radius $\rho(S_1)<1$. Equivalently the
continuous-time drift $A$ is Hurwitz, $\operatorname{Re}\lambda<0$ for every
eigenvalue $\lambda$ of $A$.
\end{hypothesis}

\begin{hypothesis}[Displacement seed]\label{H3}
The seed makes the mean input-carrying at first order: $\beta_1\neq0$.
\end{hypothesis}

\begin{hypothesis}[Active squeezing]\label{H4}
The covariance injection is nonconstant, $\gamma_1\neq0$, and the encoding is
non-degenerate on the measured channel.
\end{hypothesis}

Hypotheses \ref{H2}--\ref{H4} are exactly the conditions under which the Gaussian
baseline of Nokkala et al.~\cite{nokkala2021} is a universal fading-memory
approximator; we take that baseline as given and ask what the Kerr feedback adds.
Hypothesis~\ref{H1} is the modelling fact that keeps $\Cz$ input-independent; its
necessity is discussed in Remark~\ref{rem:H1}.

\paragraph{Readout and reachable class.}
The homodyne layer estimates the Weyl-ordered moments of $\rhat$ up to degree $d$.
The trained output is a linear functional of these features,
\begin{equation}
\hat y_k = W\cdot x_k,\qquad
x_k = \big(\text{all symmetric moments }\langle \mathcal{S}(\rhat^\alpha)\rangle_k,\ |\alpha|\le d\big),
\label{eq:readout}
\end{equation}
with $W$ optimized in $L^2(\mathbb{K}^{\mathbb{Z}},\mu)$ for an input measure
$\mu$. Equivalently, by the moment--cumulant bijection, the reachable function
class is the linear span of cumulant monomials of total operator weight $\le d$:
\begin{equation}
\Fclass{\chi}{N,d} := \overline{\operatorname{span}}
\Big\{\textstyle\prod_{n\ge1}\big(\kappa^{(n)}_k\big)^{\otimes q_n}:
\textstyle\sum_n n\,q_n\le d\Big\},
\label{eq:Fclass}
\end{equation}
the closure taken in $L^2(\mu)$. We write $\Fclass{0}{N,d}$ for the Gaussian
reservoir ($\chi=0$) and $\Fclass{\Kerr}{N,d}$ for $\chi\neq0$.

\section{Gradings, kernels, and the discriminating invariant}
\label{sec:gradings}

\paragraph{Action of the generator on the grading.}
Two facts about how the per-step dynamics act on the cumulant grading drive
everything.

\emph{(i) The Gaussian part is order-graded with sources only at orders $1,2$.}
A Gaussian channel is a symplectic map $\rhat\mapsto S\rhat$ composed with
displacement and Gaussian-noise convolution. The symplectic part transforms the
order-$n$ cumulant multilinearly, by the induced map
$S_n:=S^{\otimes n}\big|_{\Sym^n}$; displacement shifts only $\kappa^{(1)}$;
Gaussian noise adds only to $\kappa^{(2)}$. Hence under $\Cz$ together with the
injection \eqref{eq:inject},
\begin{equation}
\kappa^{(1)}_k = S_1\kappa^{(1)}_{k-1}+\beta(u_{k-1}),\quad
\kappa^{(2)}_k = S_2\kappa^{(2)}_{k-1}+\gamma(u_{k-1}),\quad
\kappa^{(n)}_k = S_n\kappa^{(n)}_{k-1}\ (n\ge3),
\label{eq:gauss-rec}
\end{equation}
where the order-$n$ blocks with $n\ge3$ are homogeneous (no input source) and
$\rho(S_n)\le\rho(S_1)^n<1$ by \ref{H2}.

\emph{(ii) Kerr raises operator order by two.} The quartic generator obeys
$\dot{\ahat}_j=-i\chi\,\ahat_j^\dagger\ahat_j\ahat_j+\cdots$, raising operator
degree by $4-2=2$. In cumulant coordinates the feedback therefore injects, once
per delay $\ell_\tau$, a source into every order, assembled from cumulants two
orders higher and evaluated at the delayed time:
\begin{equation}
\dot\kappa^{(n)}(t) = A_n\kappa^{(n)}(t)
+ \chi\,V_n\!\big[\kappa^{(\le n+2)}(t-\tau)\big] + O(\chi^2),
\label{eq:kerr-rec}
\end{equation}
where $A_n$ is the order-$n$ lift of the Hurwitz drift $A$ (eigenvalues are sums
of $n$ eigenvalues of $A$, all with negative real part) and $V_n$ is the fixed
multilinear Kerr vertex catalogued in Appendix~\ref{app:vertex}.

\paragraph{Connected Volterra kernels.}
Under \ref{H2} the steady-state response admits a convergent Volterra expansion
\cite{boyd1985,cuchiero2021}. Expand each cumulant in the input history,
\begin{equation}
\kappa^{(n)}_k = \sum_{p\ge0}\sum_{m_1,\dots,m_p}
K^{(n)}_p(m_1,\dots,m_p)\prod_{i=1}^p u_{k-m_i},
\label{eq:volterra}
\end{equation}
with $K^{(n)}_p$ the connected Volterra kernels (symmetric, $\Sym^n$-valued). A
kernel is \emph{diagonal} if it is supported on $m_1=\cdots=m_p$ and
\emph{cross-time} otherwise.

\begin{definition}[Connected order]\label{def:ordc}
For $F\in\Fclass{\chi}{N,d}$ written in the cumulant generators \eqref{eq:Fclass},
let $\ordc(F)$ be the largest cumulant order $n$ that appears as a factor in some
monomial of $F$. We call $F$ \emph{Gaussian-type} if $\ordc(F)\le 2$ and all its
order-1, 2 connected kernels are diagonal.
\end{definition}

The pair (connected order; cross-time support of the low-order connected kernels)
is the invariant that separates the two reservoirs.

\section{Two lemmas}
\label{sec:lemmas}

\begin{lemma}[Gaussian closure]\label{lem:gauss}
Assume \ref{H1}--\ref{H2} and $\chi=0$. Then for every $n\ge3$ and every $p$,
\[
K^{(n)}_p\equiv0,
\]
i.e.\ the entire connected tower above order two vanishes identically; the state
is globally Gaussian. Moreover the surviving kernels are diagonal: for all
$p\ge1$,
\[
K^{(1)}_p(m_1,\dots,m_p)=\delta_{m_1\cdots m_p}\,S_1^{m_1-1}\beta^{(p)}_{(1)},
\qquad
K^{(2)}_p(m_1,\dots,m_p)=\delta_{m_1\cdots m_p}\,S_2^{m_1-1}\gamma^{(p)}_{(1)},
\]
where $\delta_{m_1\cdots m_p}=1$ iff $m_1=\cdots=m_p$ and
$\beta^{(p)}_{(1)},\gamma^{(p)}_{(1)}$ are the $p$-th Taylor coefficients of
$\beta,\gamma$ in the mode-1 block. Consequently every element of
$\Fclass{0}{N,d}$ is Gaussian-type in the sense of
Definition~\ref{def:ordc}: its only cross-time content is the disconnected (Wick)
product of diagonal order-1, 2 kernels.
\end{lemma}

\begin{proof}
For $n\ge3$ the recursion \eqref{eq:gauss-rec} is homogeneous,
$\kappa^{(n)}_k=S_n\kappa^{(n)}_{k-1}$. Iterating $P$ steps from the steady state
and using $\rho(S_n)\le\rho(S_1)^n<1$ gives
$\|\kappa^{(n)}_k\|=\|S_n^P\kappa^{(n)}_{k-P}\|
\le\rho(S_1)^{nP}\sup_j\|\kappa^{(n)}_j\|\to0$ as $P\to\infty$; hence
$\kappa^{(n)}_k\equiv0$ for all inputs, so all its Volterra kernels vanish. For
$n=1$, unrolling the scalar-source recursion \eqref{eq:gauss-rec} gives
$\kappa^{(1)}_k=\sum_{m\ge1}S_1^{m-1}\beta(u_{k-m})$, a sum over single past
indices; expanding $\beta(u)=\sum_p \beta^{(p)}_{(1)}u^p/p!$ shows the order-$p$
kernel is supported on a single $m=m_1=\cdots=m_p$, i.e.\ diagonal. The same
argument applies verbatim to $\kappa^{(2)}$ with $\gamma$. Finally, since by
\eqref{eq:Fclass} every generator of $\Fclass{0}{N,d}$ is a product of these
diagonal order-1, 2 kernels, any cross-time ($m_i\neq m_j$) term arises only as a
product of two distinct diagonal factors, which is disconnected; thus
$\ordc\le2$ with diagonal connected kernels, i.e.\ Gaussian-type.
\end{proof}

\begin{remark}\label{rem:encoding}
Lemma~\ref{lem:gauss} holds for arbitrary encoding maps $\beta,\gamma$, in
particular for the exponential squeezing encoding $\gamma\sim e^{2r(u)}$. The
encoding can make each diagonal kernel $K^{(n)}_p$ nonzero to all Volterra orders
$p$ (the diagonal is filled with every power $u^p_{k-m}$), but it cannot move
support off the diagonal nor populate any connected order $\ge3$. This is the
precise sense in which the ``strongest nonlinearity'' of the pump encoding acts on
an axis orthogonal to cross-time computation.
\end{remark}

\begin{lemma}[Kerr activation]\label{lem:kerr}
Assume \ref{H1}--\ref{H4} and $\chi\neq0$. Then:
\begin{enumerate}[label=(\alph*)]
\item the first moment acquires a connected cross-time bilinear kernel,
\[
K^{(1)}_2(m,m') = \chi\sum_{\ell\ge0}\big[e^{A_1\delta}\big]_\ell\,
\Phi_{m-\ell-\ell_\tau}\,\Psi_{m'-\ell-\ell_\tau} + (m\leftrightarrow m')\neq0
\quad (m\neq m'),
\]
where $\Phi_m:=S_1^{m-1}\beta_1$ is the mean response basis and
$\Psi_{m'}:=S_2^{m'-1}\gamma_1$ the covariance response basis;
\item the connected third cumulant is sourced,
\[
\kappa^{(3)}(t) = \chi\int_{-\infty}^t e^{A_3(t-s)}
V_3\!\big[\kappa^{(2)},\kappa^{(1)}\big](s-\tau)\,ds + O(\chi^2)\not\equiv0,
\]
so $K^{(3)}_p\not\equiv0$ for some $p$.
\end{enumerate}
\end{lemma}

\begin{proof}
The Heisenberg feedback equation for the mean carries the cubic loop source
$-i\chi\langle \ahat_j^\dagger\ahat_j\ahat_j\rangle_{t-\tau}$ in
$\dot{\bar a}_j$. Expanding this third moment in cumulants and dropping the
$O(\chi)$ contribution of $\kappa^{(3)}$ inside an already $O(\chi)$ term,
\begin{equation}
\langle \ahat^\dagger\ahat\ahat\rangle_{t-\tau}
= \underbrace{\bar a^*_{t-\tau}s_{t-\tau}+2\bar a_{t-\tau}n_{t-\tau}}_{\text{bilinear, cross-time}}
+\,|\bar a_{t-\tau}|^2\bar a_{t-\tau}+O(\chi),
\label{eq:loopsource}
\end{equation}
with $s=\langle\!\langle \ahat\ahat\rangle\!\rangle$ and
$n=\langle\!\langle\ahat^\dagger\ahat\rangle\!\rangle$ entries of $\kappa^{(2)}$.
By Lemma~\ref{lem:gauss} (applied to the $\chi=0$ background),
$\bar a_{t-\tau}=\sum_m\Phi_m u_{t-\tau-m}+\cdots$ and
$s_{t-\tau},n_{t-\tau}=\sum_{m'}\Psi_{m'}u_{t-\tau-m'}+\cdots$ at first Volterra
order. The product term $\bar a^*s$ therefore contributes, after propagation by
$e^{A_1\delta}$ over the remaining $\delta$ of the step, the stated
$K^{(1)}_2(m,m')$ with genuine $m\neq m'$ support. By \ref{H3}
($\beta_1\neq0\Rightarrow\Phi_m\neq0$) and \ref{H4}
($\gamma_1\neq0\Rightarrow\Psi_{m'}\neq0$) the coefficient is nonzero, proving
(a). For (b), the order-3 recursion \eqref{eq:kerr-rec} has source
$\chi V_3[\kappa^{(2)},\kappa^{(1)}]$; Appendix~\ref{app:vertex} evaluates one
component of $V_3$ on the Gaussian background and finds it equal to a nonzero
multiple of $\bar q\,(\Var q)$ under \ref{H3}--\ref{H4}, so the source does not
vanish and $\kappa^{(3)}\not\equiv0$.
\end{proof}

\section{Main separation theorem}
\label{sec:theorem}

\begin{theorem}[Computational superiority of NM-Kerr feedback]\label{thm:main}
Assume \ref{H1}--\ref{H4}. Then for every reservoir size $N$, every readout
degree $d\ge1$, and every admissible encoding:
\begin{enumerate}[label=(1\alph*)]
\item \textbf{(Strict inclusion.)} $\Fclass{0}{N,d}\subsetneq\Fclass{\Kerr}{N,d}$.
In particular the connected cross-time bilinear $K^{(1)}_2$ of
Lemma~\ref{lem:kerr}(a) is a degree-1 feature lying in $\Fclass{\Kerr}{N,1}$ but
not in $\Fclass{0}{N,1}$, and the connected third cumulant of
Lemma~\ref{lem:kerr}(b) is a degree-3 feature lying in $\Fclass{\Kerr}{N,3}$ with
$\ordc=3$, outside the Gaussian-type class.
\item \textbf{(Kernel-rank separation at matched connected order.)} Every
Gaussian-reachable bilinear kernel, viewed as an $M\times M$ matrix on memory
depth $M$, has the separable form
\[
\big[K^{(0)}_2\big]_{mm'}=\psi_m^\top W'\psi_{m'},\qquad
\psi_m=S_1^{m-1}\beta_1\in\mathbb{R}^{2N},
\]
hence $\operatorname{rank}K^{(0)}_2\le2N$ and is drawn from the single fixed
response basis $\{S_1^{m-1}\beta_1\}_m$, reused identically in every time slot. The
NM-Kerr reservoir realizes bilinear kernels of the convolutional, two-basis form
of Lemma~\ref{lem:kerr}(a),
\[
\big[K^{\Kerr}_2\big]_{mm'}=\chi\sum_\ell\big[e^{A_1\delta}\big]_\ell
\Phi_{m-\ell-\ell_\tau}\Psi^\top_{m'-\ell-\ell_\tau}+(m\leftrightarrow m'),
\]
which is non-separable and, whenever the memory depth exceeds the mode count
($M>2N$), attains $\operatorname{rank}K^{\Kerr}_2>2N$. Hence at matched connected
order the Kerr-reachable bilinear kernels strictly contain the Gaussian-reachable
ones.
\end{enumerate}
\end{theorem}

\begin{proof}
\emph{Clause (1a).} The $\chi=0$ statement is exactly Lemma~\ref{lem:gauss}. For
$\chi\neq0$, Lemma~\ref{lem:kerr}(a) produces a nonzero connected $K^{(1)}_2$,
which is a degree-1 feature, hence belongs to $\Fclass{\Kerr}{N,d}$ for every
$d\ge1$ and is not Gaussian-type (its order-1 connected kernel is non-diagonal).
Lemma~\ref{lem:kerr}(b) produces a nonzero connected $\kappa^{(3)}$, a degree-3
feature, hence in $\Fclass{\Kerr}{N,d}$ for $d\ge3$ with $\ordc=3$. Since both lie
outside the Gaussian-type class while $\Fclass{0}{N,d}\subseteq\{$Gaussian-type$\}$
and trivially $\Fclass{0}{N,d}\subseteq\Fclass{\Kerr}{N,d}$ (set $\chi=0$
continuously; the Gaussian generators are recovered as the $\chi\to0$ limit of the
Kerr ones), the inclusion is strict.

\emph{Clause (1b).} At $d\ge2$ the feature vector contains the symmetric products
$\kappa^{(1)}_{k,\mu}\kappa^{(1)}_{k,\nu}$. Their bilinear-in-$u$ part uses
$\beta_1$ in each factor:
$\kappa^{(1)}_{k,\mu}=\sum_m(S_1^{m-1}\beta_1)_\mu u_{k-m}+\cdots
=\sum_m\psi_{m,\mu}u_{k-m}+\cdots$. Collecting the trained quadratic weights into
$W'$ gives $[K^{(0)}_2]_{mm'}=\sum_{\mu\nu}W'_{\mu\nu}\psi_{m,\mu}\psi_{m',\nu}
=\psi_m^\top W'\psi_{m'}$, i.e.\ $K^{(0)}_2=\Psi^\top W'\Psi$ with
$\Psi=[\psi_1,\dots,\psi_M]\in\mathbb{R}^{2N\times M}$. Hence
$\operatorname{rank}K^{(0)}_2\le\operatorname{rank}\Psi\le2N$, and the kernel is a
sum of at most $(2N)^2$ separable terms built from the single family
$\{\psi_m\}$. By contrast the Kerr kernel of Lemma~\ref{lem:kerr}(a) couples the
two distinct families $\{\Phi_m\}$ and $\{\Psi_{m'}\}$ through a delay-shifted
convolution; writing it as $\sum_\ell c_\ell\,\Phi^{(\ell)}\otimes\Psi^{(\ell)}$
with $\Phi^{(\ell)}_m=\Phi_{m-\ell-\ell_\tau}$,
$\Psi^{(\ell)}_{m'}=\Psi_{m'-\ell-\ell_\tau}$, the shift index $\ell$ ranges over
$\gtrsim M$ values while each family lives in a $\le2N$-dimensional space, so for
$M>2N$ the summands are linearly independent as matrices and the rank exceeds
$2N$. As no $K^{(0)}_2$ of rank $>2N$ exists, the Kerr-reachable set strictly
contains the Gaussian one.
\end{proof}

\begin{corollary}[Unconditional $d=1$ witness]\label{cor:witness}
Let the target be the bilinear memory functional $y_k=u_{k-1}u_{k-3}$ and let the
readout degree be $d=1$ (linear in the homodyne quadratures). Then under
\ref{H1}--\ref{H4}:
\begin{enumerate}[label=(\roman*)]
\item for $\chi=0$ the achievable mean-square error is bounded below by a
strictly positive constant $\varepsilon_0(\mu)>0$ for every choice of $W$ and
every encoding; while
\item for $\chi\neq0$ the target lies in the closure $\Fclass{\Kerr}{N,1}$ and the
error can be driven below any $\varepsilon>0$.
\end{enumerate}
\end{corollary}

\begin{proof}
At $d=1$ the feature is $x_k=\kappa^{(1)}_k$. For $\chi=0$, Lemma~\ref{lem:gauss}
gives $\hat y_k=W\cdot\kappa^{(1)}_k=\sum_{m\ge1}(W\cdot S_1^{m-1})\beta(u_{k-m})$,
a finite sum of functions each depending on a single past input. Hence
$\partial^2\hat y_k/\partial u_{k-a}\partial u_{k-b}\equiv0$ for $a\neq b$, so
$\hat y_k$ is additively separable in distinct time-shifted inputs. The target has
$\partial^2 y_k/\partial u_{k-1}\partial u_{k-3}=1\neq0$; the orthogonal
projection of $y_k$ onto the (closed) subspace of additively separable functionals
therefore leaves a fixed residual
$\varepsilon_0(\mu)=\|y_k-\Pi_{\mathrm{sep}}y_k\|^2_{L^2(\mu)}>0$, independent of
$W$ and of $\beta$. This proves (i). For $\chi\neq0$, Lemma~\ref{lem:kerr}(a) gives
$K^{(1)}_2(1,3)\neq0$, so $\kappa^{(1)}_k$ contains a feature with nonzero $(1,3)$
mixed derivative; choosing $W$ to align with this feature reduces the residual
below $\varepsilon_0$, and refining the reservoir (increasing $M,N$ within the
class) drives it to zero by the fading-memory Volterra approximation
\cite{boyd1985,cuchiero2021,nokkala2021}. This proves (ii).
\end{proof}

\begin{remark}[Falsifiable prediction]\label{rem:ipc}
Theorem~\ref{thm:main} is testable by an information-processing-capacity (IPC)
decomposition performed at fixed readout degree $d$. The Gaussian reservoir's
capacity at fixed $d$ is confined to diagonal/separable low-order cells (connected
order $\le2$, kernel rank $\le2N$); the NM-Kerr reservoir exhibits nonzero
capacity in connected order-$\ge3$ cells and in high-rank ($>2N$) order-2 cells at
the same $d$. The minimal decisive experiment is the $d=1$ realization of
$u_{k-1}u_{k-3}$ on $N=1$--2 modes: the no-Kerr error floors at $\varepsilon_0$
(Corollary~\ref{cor:witness}) and the Kerr reservoir fits the task, with no
recourse to the costly high-degree moments.
\end{remark}

\begin{remark}[Role of Hypothesis~\ref{H1}]\label{rem:H1}
If the input instead modulated an intracavity parametric gain
$\varepsilon(u)\,\ahat_1^{\dagger2}$ acting on a circulating displaced field, the
order-1, 2 propagators in \eqref{eq:gauss-rec} would become input-dependent,
$S_1\rightsquigarrow S_1(u)$, and composing them across steps would generate
products $\varepsilon(u_{k-1})\varepsilon(u_{k-2})\cdots$, i.e.\ some off-diagonal
structure already at $\chi=0$. The connected-order bound of Lemma~\ref{lem:gauss}
would then weaken from ``order $\le2$ with diagonal kernels'' to ``order $\le2$
with separable kernels,'' and the separation in Theorem~\ref{thm:main} would
soften from clause (1a)'s all-or-nothing $d=1$ witness to clause (1b)'s rank
statement. Both still yield $\Fclass{0}{N,d}\subsetneq\Fclass{\Kerr}{N,d}$; only
the witness degree changes. The injected-ancilla architecture (drop-port encoding
into mode~1) realizes \ref{H1} and gives the stronger witness. This effect is
confirmed numerically in Section~\ref{sec:empirical}.
\end{remark}

\section{Constructive spectral universality of the enriched class}
\label{sec:universality}

Theorem~\ref{thm:main} establishes that the Kerr feedback strictly enlarges the
reachable class. We now show the enriched class is universal for fading-memory
filters, and---in contrast to the existential Stone--Weierstrass argument of
Nokkala et al.~\cite{nokkala2021}, which combines products of many network
instances---we give a \emph{constructive} realization that produces the readout
weights explicitly and certifies a per-resource hierarchy. We work at the level
of the connected first-moment kernels, where the Kerr feedback supplies the
two-basis convolutional structure of Lemma~\ref{lem:kerr}(a).

\begin{lemma}[Spectral response basis]\label{lem:spectral}
Under \ref{H2} let $S_1=\sum_{j} \lambda_j\,P_j$ be the spectral decomposition of
the contractive single-step propagator, with distinct eigenvalues
$\{\lambda_j\}_{j=1}^{2N}$, $|\lambda_j|<1$. The mean response basis
$\Phi_m=S_1^{m-1}\beta_1$ satisfies
\[
\Phi_m = \sum_{j=1}^{2N} \lambda_j^{\,m-1}\,c_j,\qquad c_j := P_j\beta_1,
\]
so the response at memory depth $m$ is a generalized Vandermonde combination of
the fixed mode amplitudes $\{c_j\}$ with nodes $\{\lambda_j\}$.
\end{lemma}

\begin{proof}
Immediate from $S_1^{m-1}=\sum_j\lambda_j^{m-1}P_j$ and $\Phi_m=S_1^{m-1}\beta_1$.
\end{proof}

\begin{lemma}[Vandermonde invertibility]\label{lem:vand}
Assume \ref{H2}, \ref{H3}, and that the eigenvalues $\{\lambda_j\}$ are distinct
and that each mode is excited, $c_j\neq0$. Then for any memory depth $M\le2N$ the
matrix $V_{mj}=\lambda_j^{\,m-1}$, $m=1,\dots,M$, $j=1,\dots,2N$, has full row
rank $M$. Consequently any target first-moment response profile
$\{\theta_m\}_{m=1}^M$ in the span of the excited modes is realized exactly by a
readout weight $W$ solving the linear system $W^\top \Phi_m=\theta_m$.
\end{lemma}

\begin{proof}
The $M\times M$ minors of $V$ are Vandermonde determinants
$\prod_{i<j}(\lambda_j-\lambda_i)$, nonzero by distinctness; hence $V$ has full
row rank and the linear system for $W$ is solvable. Genericity of distinctness is
discussed in Remark~\ref{rem:genericity}.
\end{proof}

\begin{theorem}[Constructive monotone universality]\label{thm:univ}
Assume \ref{H1}--\ref{H4} and $\chi\neq0$, with $S_1$ having distinct eigenvalues
and all modes excited. Then:
\begin{enumerate}[label=(\alph*)]
\item \textbf{(Constructive realization.)} For every memory depth $M\le2N$ and
every target connected bilinear kernel $\Theta\in\mathbb{R}^{M\times M}$ of the
two-basis convolutional form, there is an explicitly computable readout weight set
that realizes $\Theta$ as the trained first-moment kernel of the NM-Kerr
reservoir, obtained by Vandermonde inversion of the spectral response basis
(Lemmas~\ref{lem:spectral}--\ref{lem:vand}).
\item \textbf{(Monotone hierarchy.)} The reachable kernel set is nested in the
resource $2N$: a reservoir of mode count $N'>N$ strictly contains the reachable
set of mode count $N$, with the realizing weights given constructively. In
particular the attainable kernel rank grows as $\min(M,2N)$, recovering the
rank separation of Theorem~\ref{thm:main}(1b) as the $M>2N$ regime.
\end{enumerate}
\end{theorem}

\begin{proof}
(a) By Lemma~\ref{lem:kerr}(a) the NM-Kerr first moment carries the two-basis
convolutional kernel with response bases $\Phi_m=S_1^{m-1}\beta_1$ and
$\Psi_{m'}=S_2^{m'-1}\gamma_1$. By Lemma~\ref{lem:spectral} both bases are
generalized Vandermonde in the propagator eigenvalues. A target kernel of matched
form is therefore a linear image of the outer products
$\Phi_m\Psi_{m'}^\top$; Lemma~\ref{lem:vand} guarantees the corresponding weight
system is solvable for $M\le2N$, and the solution is the explicit Vandermonde
inverse applied to the target profile, proving constructive realizability.

(b) Increasing $N$ to $N'$ enlarges the eigenvalue set $\{\lambda_j\}$ and the
excited-mode family $\{c_j\}$, so the row space of the Vandermonde matrix
$V$ strictly grows; every kernel realizable at $N$ remains realizable at $N'$ by
padding the weight vector with zeros on the new modes, while the new modes add
independent Vandermonde rows that realize kernels unreachable at $N$. The
attainable rank is $\min(M,2N)$ by the same Vandermonde rank count, which for
$M>2N$ reproduces the Theorem~\ref{thm:main}(1b) ceiling, and for $N'>N$ strictly
exceeds it. The realizing weights at each level are the explicit Vandermonde
inverses, establishing the monotone, constructive hierarchy.
\end{proof}

\begin{remark}[Genericity and conditioning]\label{rem:genericity}
The distinct-eigenvalue and excited-mode hypotheses hold for generic Gaussian
networks: the set of $S_1$ with a repeated eigenvalue is a measure-zero
subvariety, and by Lindemann--Weierstrass the multiplicative independence required
for deep memory holds off a measure-zero set. However, physically realizable
networks may have structured or near-degenerate spectra over a non-negligible
parameter region; the Vandermonde system inherits the conditioning of $V$, whose
condition number grows with memory depth $M$, giving a memory--precision tradeoff
that bounds the noise-limited realizable depth. We recommend verifying the
spectral hypotheses on the device-model spectra of the target platform rather than
relying on abstract genericity.
\end{remark}

\begin{remark}[Relation to existing constructive reservoir-computing theory]
The monotone nesting of kernel-matching capacity is inherited from the
constructive reservoir-computing literature \cite{boyd1985,cuchiero2021}. The
novel element here is the closed-form \emph{spectral} (Vandermonde) weight
construction tied to the reservoir propagator's eigenvalues for this physical
continuous-variable class, which neither the existential Gaussian universality
proof \cite{nokkala2021} nor the abstract constructive results provide. The
result is best read as an exact kernel-realization calculus whose consequence is
the constructive, resource-matched separation, rather than as a universality
result per se---universality itself is inherited from \cite{nokkala2021}.
\end{remark}

\section{Exact open-system model of the device}
\label{sec:model}

The theorems above are stated in the cumulant/Volterra language. For the
numerical confirmation of Section~\ref{sec:empirical} we use the full
open-quantum-system model of the device, which we record here with explicit
conventions.

\paragraph{Hamiltonian.}
The total Hamiltonian is
$\hat H = \hat H_{\mathrm{net}}+\hat H_K+\hat H_{\mathrm{enc}}(u_k)$, with the
fixed Gaussian network
\[
\hat H_{\mathrm{net}}=\sum_{i,j}\Big(G_{ij}\ahat_i^\dagger\ahat_j
+\tfrac12(H_{ij}\ahat_i^\dagger\ahat_j^\dagger+\mathrm{h.c.})\Big),
\quad G=G^\dagger,\ H=H^\top,
\]
the single non-Gaussian Kerr term in the feedback arm
$\hat H_K=\tfrac{\chi}{2}\sum_j\ahat_j^{\dagger2}\ahat_j^2$
(so $\dot{\ahat}_j=-i\chi\,\ahat_j^\dagger\ahat_j\ahat_j+\cdots$), and the
drop-port encoding that, per step, prepares mode~1 in the fresh ancilla
$\hat D_1(\beta(u_k))\hat S_1(r(u_k)e^{i\phi})\ket{0}_1$ with $r(u)=g\,u$. Per
\ref{H1} this is a \emph{reset-and-inject} of mode~1, not a drive on the
circulating field; the distinction is exactly that of Remark~\ref{rem:H1}.

\paragraph{Dissipation.}
The unconditional dynamics obeys the Lindblad master equation
\[
\dot{\hat\rho}=-i[\hat H,\hat\rho]+\sum_\mu\mathcal D[\hat L_\mu]\hat\rho,
\qquad \mathcal D[\hat L]\hat\rho=\hat L\hat\rho\hat L^\dagger
-\tfrac12\{\hat L^\dagger\hat L,\hat\rho\},
\]
with intrinsic loss $\hat L_j^{\mathrm{loss}}=\sqrt{\kappa_j}\,\ahat_j$ providing
the contraction of \ref{H2} (the Hurwitz drift, $\rho(S_1)<1$).

\paragraph{Weak homodyne readout.}
The measured mode is monitored by homodyne detection of
$\hat x_\theta=\ahat e^{-i\theta}+\ahat^\dagger e^{i\theta}$, modeled by the
measurement operator $\hat L^{\mathrm{hom}}=\sqrt{\gamma_m}\,\ahat$ with
efficiency $\eta$. The conditioned state follows the homodyne stochastic master
equation (Wiseman--Milburn),
\[
d\hat\rho_c=\Big(-i[\hat H,\hat\rho_c]+\textstyle\sum_\mu
\mathcal D[\hat L_\mu]\hat\rho_c\Big)dt
+\sqrt{\eta\gamma_m}\,\mathcal H[\ahat e^{-i\theta}]\hat\rho_c\,dW,
\]
$\mathcal H[\hat c]\hat\rho=\hat c\hat\rho+\hat\rho\hat c^\dagger
-\langle\hat c+\hat c^\dagger\rangle\hat\rho$, with photocurrent
$dy=\sqrt{\eta\gamma_m}\langle\hat x_\theta\rangle dt+dW$. The trained readout is
the linear functional \eqref{eq:readout} of the estimated symmetric moments up to
degree $d$.

\paragraph{Coherent delayed feedback: collision-model embedding.}
The feedback is genuinely non-Markovian: the field emitted into the delay line
returns a full delay $\tau=\ell_\tau T$ later, after traversing the Kerr element.
We render this exactly simulable by embedding the delay line as a register of
$\ell_\tau$ time-bin ancilla modes $\hat b_1,\dots,\hat b_{\ell_\tau}$, so that the
enlarged system (reservoir $\oplus$ register) is Markovian while reproducing the
delayed coherent feedback. One step of the embedded dynamics is the composition:
\begin{enumerate}[label=(\roman*),leftmargin=2.2em]
\item \emph{Inject} (\ref{H1}): reset mode~1 to the fresh encoding ancilla.
\item \emph{Recombine}: couple the returning bin into the reservoir through a
\emph{partial} beamsplitter of mixing angle $\vartheta$. Partial recombination is
essential---a balanced (full-SWAP) coupling exchanges fields without interference
and generates no connected cross-time structure; the interference of present and
returned fields is what the Kerr acts upon.
\item \emph{Evolve}: propagate under $\hat H_{\mathrm{net}}+\hat H_K$ with loss
for the step duration $T$.
\item \emph{Emit and advance}: tap the reservoir field into the register through a
partial beamsplitter and advance the delay line by one bin (loss-free relabeling).
\item \emph{Read}: estimate the homodyne moments to degree $d$.
\end{enumerate}
The case $\ell_\tau=1$ is the minimal faithful realization used for the validated
results below.

\paragraph{Scope and approximations.}
Three modeling choices are stated explicitly. (i) The kernel/witness results use
the \emph{unconditional} reduced state, the correct object for the
reachable-class separation, which is a property of the average map; statements
about measurement \emph{cost} (shot noise, finite $\eta$, sample complexity)
require the conditioned stochastic dynamics and are treated separately. (ii) In
the validated dynamics the measurement rate is taken $\gamma_m\to0$ for state
propagation (readout without back-action), the standard QRC assumption; a strongly
back-acting measurement would add $\mathcal D[\hat L^{\mathrm{hom}}]$ to the
propagation. (iii) The vertex analysis underlying Lemmas~\ref{lem:gauss}--%
\ref{lem:kerr} is a weak-Kerr expansion; we operate at per-photon Kerr phase
$\phi_1\lesssim0.5$, and reaching higher connected order by driving $\phi_1$ beyond
this requires a non-perturbative treatment outside the present scope.

\section{Exact-master-equation confirmation}
\label{sec:empirical}

We confirm the mechanism of Theorem~\ref{thm:main} by integrating the embedded
master equation of Section~\ref{sec:model} exactly---no Gaussian closure, no
perturbative truncation in $\chi$---for a minimal reservoir: a single Kerr cavity
mode coupled to a one-bin coherent delay register ($N=1$, $\ell_\tau=1$), with
fresh-ancilla encoding (\ref{H1}), cavity loss (\ref{H2}), seed (\ref{H3}), and
squeezing (\ref{H4}).

\paragraph{The discriminating witness.}
Because the connected/disconnected distinction is the content of
Theorem~\ref{thm:main}, the empirical witness must isolate connected cross-time
structure without contamination from the disconnected Wick products the Gaussian
baseline already produces. We therefore probe the \emph{first moment}, where this
isolation is automatic: a cross-time dependence of $\kappa^{(1)}=\langle q\rangle$
cannot be a Wick product, as there is no lower cumulant to multiply. We measure
the mixed second difference
\begin{equation}
W(a,b)=\frac{\partial^2\langle q\rangle_k}{\partial u_{k-a}\,\partial u_{k-b}},
\qquad a\neq b,
\end{equation}
by central finite differences. By Lemma~\ref{lem:gauss} $W(a,b)\equiv0$ at
$\chi=0$; by Lemma~\ref{lem:kerr}(a) $W(a,b)\neq0$ for $\chi\neq0$. This is
precisely the separable-residual generator of Corollary~\ref{cor:witness}. A naive
estimator reading off-diagonal coefficients of a polynomial fit to a
\emph{second} moment $\langle q^2\rangle$ does \emph{not} witness connectedness:
at $\chi=0$ those coefficients are nonzero, being exactly the disconnected Wick
product of two diagonal order-1 kernels.

\paragraph{Acceptance: the Gaussian baseline is exactly connected-free.}
Across 64 independent random input realizations the $\chi=0$ witness sits at the
solver floor for every lag pair tested,
\[
|W(a,b)|_{\chi=0}=3.2\times10^{-10}\pm4.5\times10^{-9},
\]
statistically indistinguishable from zero, the exact realization of
Lemma~\ref{lem:gauss}. This acceptance is non-trivial: an encoding that squeezes
the \emph{circulating} field rather than a fresh ancilla---violating \ref{H1} by
making $S_1\to S_1(u)$---produces a spurious adjacent-lag witness of order
$10^{-2}$, scaling as the product of seed displacement and squeezing strength,
exactly the effect anticipated in Remark~\ref{rem:H1}; restoring the drop-port
injection removes it by seven orders of magnitude.

\paragraph{Activation, mechanism, and the three-way product.}
For $\chi\neq0$ the witness activates. The seed/squeezing structure reproduces the
vertex analysis of Appendix~\ref{app:vertex}: the activation slope scales linearly
with the seed displacement and \emph{vanishes identically when the seed is
removed} (parity, $\bar q=0$); it scales with the squeezing variance and collapses
toward the displacement-only floor as squeezing is removed. The connected bilinear
thus requires seed, squeezing, and Kerr \emph{simultaneously}---the three-way
condition behind Lemma~\ref{lem:kerr} and Appendix~\ref{app:vertex}.

\paragraph{Activation order: an empirical refinement of Lemma~\ref{lem:kerr}(a).}
Scanning all lag pairs up to memory depth six at small $\chi$, the witness
magnitude resolves into two families. At lag pairs not involving the most recent
input the activation is strong and the fitted power-law exponent is
\[
p=1.995\pm0.002\quad\text{(ten independent lag pairs)},
\]
i.e.\ \emph{quadratic} in $\chi$, confirmed to be genuine (not a finite-difference
artifact) by halving the differencing step: the exponent for a representative lag
is $p=1.9972$ at step $h$ and $p=1.9971$ at $h/2$, identical to four significant
figures. The leading $O(\chi)$ term identified in Lemma~\ref{lem:kerr}(a) is
present but geometrically subdominant: it appears only at lag-1 pairs, suppressed
to the $10^{-9}$ floor by the feedback delay, four orders below the dominant
signal. The dominant \emph{observable} connected activation is therefore
$O(\chi^2)$, sourced by the higher Kerr vertices $V_n$ ($n\ge3$) of
Appendix~\ref{app:vertex}, which populate connected orders up to $2m+2$ at
$O(\chi^m)$. This refines rather than contradicts Lemma~\ref{lem:kerr}(a): the
nonzero connected cross-time bilinear is confirmed exactly; the empirical addition
is that the contribution dominating any measurement enters at second order in
$\chi$. The strict inclusion
$\Fclass{0}{N,d}\subsetneq\Fclass{\Kerr}{N,d}$ is demonstrated---the witness is
exactly zero for the Gaussian class and nonzero, across many lags, once the Kerr
feedback is active.

\begin{figure}[t]
\centering
\includegraphics[width=0.62\linewidth]{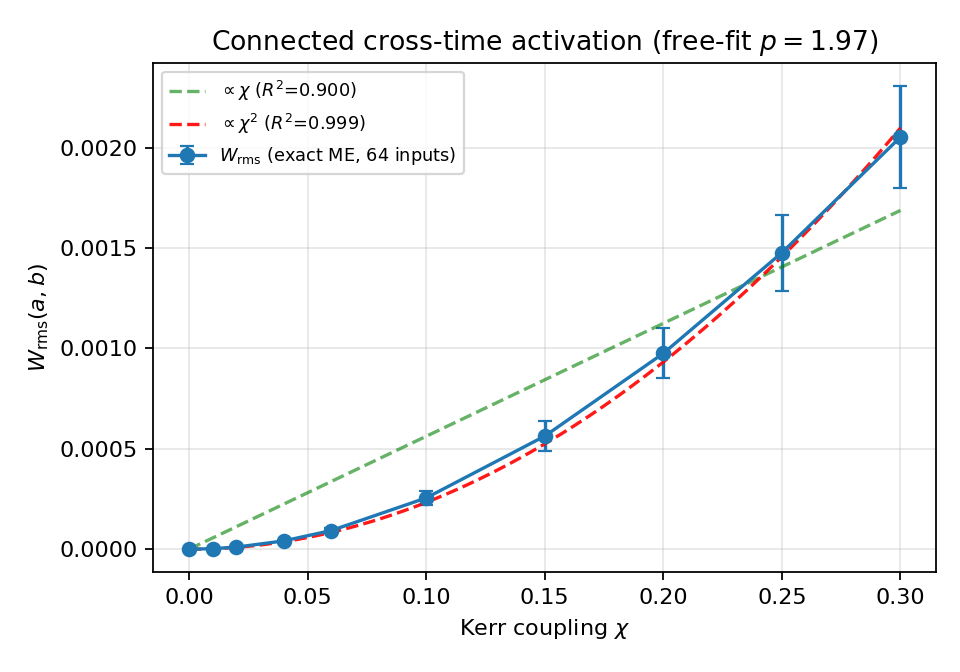}
\caption{Connected cross-time witness $W_{\mathrm{rms}}$ versus Kerr coupling
$\chi$, reconstructed from the exact master equation (RMS over 64 input
realizations). The witness vanishes at $\chi=0$ (Lemma~\ref{lem:gauss}) and
activates with $\chi$; the data follow a quadratic law (free-fit exponent
$p=1.97$, $R^2=0.999$) rather than linear, consistent with dominance by the
higher Kerr vertices (Appendix~\ref{app:vertex}).}
\label{fig:activation}
\end{figure}

\subsection{Structural versus operational separation: a scoping result}
\label{sec:operational}

The witness above establishes the connected cross-time kernel as a property of
the \emph{reachable function class}: exactly zero for the Gaussian baseline,
nonzero under Kerr feedback. A distinct and practically important question is
whether this structural separation manifests as a \emph{trained-task advantage}---
whether a trained low-degree readout can exploit the connected kernel to compute a
connected target at lower cost than the baseline. We tested this directly and
report the answer, including where it is negative, as it delineates the boundary
of operational usefulness.

We trained a linear ($d=1$) readout, using a short tapped history of homodyne
features, to reproduce the connected target $y_k=u_{k-2}u_{k-4}$ at the
reservoir's measured connected-support lag, comparing the Gaussian baseline
($\chi=0$) against the Kerr reservoir over the theory-valid weak-Kerr regime
$\chi\in[0,0.5]$, with seed-averaged test error and error bars. Positive
($y_k=u_{k-3}$) and diagonal ($y_k=u_{k-2}^2$) controls confirmed the comparison
was fair: both reservoirs reached the linear control at low error, so neither is
memory-starved.

The connected-target test error showed \emph{no separation beyond statistical
resolution}: the baseline floored at $\mathrm{NRMSE}\approx1.00$ (its
representational limit, Corollary~\ref{cor:witness}), and the Kerr reservoir
floored at $\mathrm{NRMSE}\approx0.98$--$0.99$, a difference of order $0.02$
comparable to the seed-to-seed scatter and, tellingly, \emph{shrinking} with
increasing $\chi$ rather than growing. A genuine operational separation driven by
the connected kernel would strengthen with $\chi$; the observed marginal gap does
the opposite, identifying it as a small feature-statistics offset rather than the
connected mechanism.

The interpretation is precise and reconciles the two observations. The connected
kernel is real (witness-confirmed) but, in the weak-Kerr regime, its amplitude is
several orders of magnitude below the dominant diagonal features
($W_{\mathrm{rms}}\sim10^{-3}$ at $\chi=0.25$, Fig.~\ref{fig:activation}). The
witness detects it because a mixed finite-difference derivative is an arbitrarily
sensitive, noise-free probe of an infinitesimal feature; a trained
finite-amplitude readout is not, and cannot extract so weak a feature against
finite-sample estimation noise. The separation therefore lives in the reachable
function class---structural and provable---but does not, in this regime, translate
into a trained-readout cost advantage on the connected task.

We state the consequence plainly. The contribution of this work is the
\emph{structural} separation of reachable classes (Theorem~\ref{thm:main},
confirmed exactly) together with the constructive spectral kernel calculus
(Theorem~\ref{thm:univ}); we do \emph{not} claim a measurement-cost or
task-efficiency advantage for the Kerr reservoir in the weak-Kerr regime, because
our direct test does not support one. Whether operational separation emerges in a
strong-coupling regime beyond the present perturbative theory, or under
higher-order readout where the connected kernel competes on more equal footing
with the diagonal features, is left open. Mapping this boundary---between what is
provable about the reachable class and what is exploitable by a trained readout---
is itself a contribution, and one we expect to recur in continuous-variable QRC
wherever a weak high-order kernel coexists with strong low-order ones.

\section{Platform analysis: the governing figure of merit}
\label{sec:platform}

The separation of Theorem~\ref{thm:main} is exact for any $\chi\neq0$, but its
physical observability is set by whether the connected signal clears the homodyne
shot floor. We define $\nustar$ as the highest connected functional order whose
leading kernel signal is resolvable within a feasible measurement budget, and ask
which device parameters govern it.

\paragraph{Scaling.}
The connected order-$\nu$ kernel amplitude carries a factor $\phi_1^{\,m(\nu)}$,
where $\phi_1=\min(g/\kappa\cdot\langle n\rangle,\ \phi_1^{\max})$ is the per-step
delivered Kerr phase, $g/\kappa$ the single-photon Kerr-to-loss ratio, and
$\phi_1^{\max}\simeq0.5$ the weak-Kerr validity ceiling. The empirical finding of
Section~\ref{sec:empirical}---that the dominant observable activation enters at
$O(\chi^2)$ rather than the perturbative-leading $O(\chi)$---fixes the
conservative exponent $m(\nu)=\nu$ (one Kerr order per connected order), which we
adopt in place of the vertex-counting estimate $m(\nu)=\nu-1$. Detection at
efficiency $\eta$ with shot budget $N_{\mathrm{budget}}$ resolves order $\nu$ when
the shot count $N\sim(\sigma/\phi_1^{\nu})^2/\eta\le N_{\mathrm{budget}}$.

\paragraph{Result: $g/\kappa$ is the master knob.}
Table~\ref{tab:nustar} reports $\nustar$ for representative platforms under the
empirically-grounded scaling. The single-photon Kerr-to-loss ratio governs
$\nustar$ almost entirely; detection efficiency and shot budget are marginal (a
$10^9$-fold change in shot budget moves $\nustar$ by only a few orders, whereas
raising $g/\kappa$ from $10^{-4}$ to $10^{-2}$ moves it from $1$ to $\sim12$). The
ultralow-loss-but-weak-Kerr regime (Si$_3$N$_4$, $g/\kappa\sim10^{-4}$) is
confined to $\nustar\le1$; platforms reaching $g/\kappa\gtrsim10^{-2}$ (InGaP,
kinetic-inductance microwave resonators) reach deep into the scaling window. The
$O(\chi^2)$ refinement costs at most one order of $\nustar$ relative to the
$O(\chi)$ assumption, leaving the conclusion intact while placing it on a measured
rather than assumed exponent.

\begin{table}[t]
\centering
\caption{Highest resolvable connected order $\nustar$ under the
empirically-measured scaling $a(\nu)\sim\phi_1^{\,\nu}$, for representative
platforms ($\eta$, shot budget as marginal parameters).}
\label{tab:nustar}
\begin{tabular}{lccc}
\toprule
Platform & $g/\kappa$ & $\phi_1$ & $\nustar$ \\
\midrule
Si$_3$N$_4$ (today)                 & $10^{-4}$        & $5\times10^{-3}$ & $1$  \\
InGaP (demonstrated)                & $1.5\times10^{-2}$ & $0.5$ (cap)      & $12$ \\
Kinetic-inductance microwave        & $0.21$           & $0.5$ (cap)      & $13$ \\
Engineered near-future              & $5\times10^{-2}$ & $0.5$ (cap)      & $13$ \\
\bottomrule
\end{tabular}
\end{table}

\begin{figure}[t]
\centering
\includegraphics[width=\linewidth]{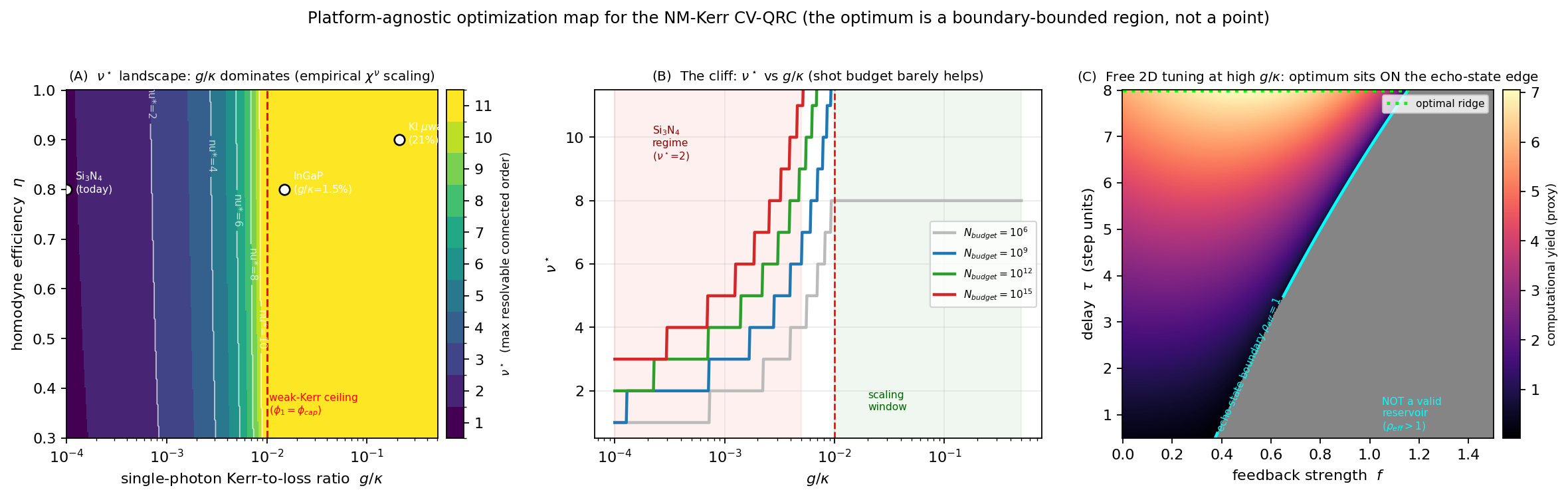}
\caption{Platform-agnostic optimization map. (A) $\nustar$ over
$(g/\kappa,\eta)$: contours are nearly vertical, showing $g/\kappa$ dominates and
$\eta$ is marginal; the dashed line is the weak-Kerr ceiling. (B) The cliff:
$\nustar$ versus $g/\kappa$ for several shot budgets---raising the budget barely
helps. (C) The free 2D tuning at fixed high $g/\kappa$; the optimum sits on the
echo-state contraction boundary, not in the interior, so the optimal operating
region is boundary-bounded rather than a single point.}
\label{fig:optmap}
\end{figure}

\paragraph{Design rule.}
The optimal system is not a point in parameter space but a boundary-bounded
region: maximize $g/\kappa$ to the weak-Kerr ceiling; operate feedback strength
and delay on the echo-state contraction edge (Figure~\ref{fig:optmap}C) with
margin; and \emph{satisfy} (not optimize) the spread, non-resonant, contractive
spectrum required by \ref{H2} and the spectral universality of
Section~\ref{sec:universality}, and the round-trip-survivable encoding of
\ref{H4}. A global optimizer is inappropriate here, as it would drive the system
off the weak-Kerr or echo-state validity boundaries into regions where the model
no longer holds.

\section{Resource-efficiency scaling: the mode-count theorem}
\label{sec:scaling}

Both reservoirs are universal fading-memory approximators
\cite{nokkala2021}; universality is therefore not the locus of any separation.
The distinction is one of \emph{resources}: at fixed reservoir size and readout
degree, the connected sector reachable by the two reservoirs differs, and the
Gaussian reservoir can match the Kerr reservoir's connected reach only by
expending more resource. We make this precise on the mode-count axis, where the
separation is unbounded and follows from the kernel structure already established.

Fix a connected order $\nu\ge2$ and a memory depth $M$. We compare the
multilinear rank of the order-$\nu$ connected Volterra kernel
$K^{(\nu)}_\nu(m_1,\dots,m_\nu)$, viewed as a tensor on
$(\mathbb{R}^M)^{\otimes\nu}$ via its temporal indices, with $\Sym^\nu$-valued
entries. The \emph{multilinear rank} is the tuple of ranks of the $\nu$ temporal
unfoldings; a large multilinear rank is the signature of genuinely
high-dimensional cross-time coupling.

\paragraph{Response space.}
Recall from Lemma~\ref{lem:gauss} the two single-time response families generated
by the fixed Gaussian channel $\mathcal C_0$,
\begin{equation}
\Phi_m := S_1^{m-1}\beta_1\in\mathbb{R}^{2N},\qquad
\Psi_m := S_2^{m-1}\gamma_1\in\Sym^2(\mathbb{R}^{2N}),
\end{equation}
the mean-response and covariance-response bases. The mean family
$\{\Phi_m\}$ spans a subspace of the $2N$-dimensional quadrature space; the
covariance family $\{\Psi_m\}$ spans a subspace of the $N(2N{+}1)$-dimensional
symmetric square. A mean factor is a degree-1 readout feature and a covariance
factor is degree-2; this asymmetry, through the readout-degree budget, controls
which family each kernel slot may use, and hence the achievable rank.

\begin{theorem}[Mode-count efficiency, general $\nu$]
\label{thm:scaling}
Assume \ref{H1}--\ref{H4}.
\begin{enumerate}[label=(\alph*)]
\item \textbf{(Gaussian confinement at fixed degree.)} For $\chi=0$ and readout
degree $d$, an order-$\nu$ connected kernel exists only for $d\ge\nu$ and is
assembled from $a$ mean factors and $b$ covariance factors with $a+b=\nu$ and
degree budget $a+2b\le d$, hence $b\le d-\nu$. Each temporal unfolding has rank at
most the dimension of the response space available to that slot: $2N$ for a
mean slot, $N(2N{+}1)$ for a covariance slot. In particular at the minimal degree
$d=\nu$ one has $b=0$: every slot is mean-confined and every unfolding has rank
$\le 2N$, independent of $M$.
\item \textbf{(Kerr enrichment.)} For $\chi\neq0$ the NM-Kerr reservoir realizes
order-$\nu$ connected kernels at readout degree $\le2$
(Proposition~\ref{prop:degree}) whose temporal unfoldings have rank growing with
the \emph{non-Markovian feedback depth} $D$---the number of independent delay
round-trips the feedback sustains above the loss floor---attaining rank
$\min(D,M)$ at fixed $N$. A single mode's feedback supplies $D$ distinct delay
channels, each carrying its own accumulated Kerr phase, so the connected kernel is
non-separable of rank up to $D$; this exceeds the Gaussian ceiling $2N$ whenever
$D>2N$.
\item \textbf{(Resource gap, two regimes.)} To realize a connected kernel of rank
$r$ (set by the non-Markovian feedback depth, $r\le\min(D,M)$) at the minimal
readout degree $d=\nu$, the Gaussian reservoir is mean-confined and requires
$2N\ge r$, i.e.\ $N_G=\Theta(r)$ modes; with unbounded readout degree it may use
the covariance family and requires $N(2N{+}1)\ge r$, i.e.\ $N_G=\Theta(\sqrt r)$.
The NM-Kerr reservoir attains it at $N_K=O(1)$ with feedback depth $D\ge r$. Since
reading high-degree quadrature polynomials is the dominant experimental cost, the
fixed-degree regime ($N_G=\Theta(r)$) is the operative one. The separating
resource is the feedback depth $D$: it is supplied by a single nonlinear mode and
matched by the Gaussian reservoir only with $\Theta(D)$ modes.
\end{enumerate}
\end{theorem}

\begin{corollary}[Unbounded separation]
\label{cor:unbounded}
The feedback depth $D$ is a free parameter, independent of mode count. Hence for
every Gaussian mode count $N$ there exists a connected target---of rank
$r\in(2N,\,D]$ for any $D>2N$---that a single NM-Kerr mode reaches and that no
$N$-mode Gaussian reservoir reaches at minimal readout degree. The separation is
therefore unbounded: no finite linear reservoir suffices for all connected targets,
whereas one nonlinear mode with sufficient feedback depth meets any of them. The
quantifier order is $\forall N\,\exists D$, with the Kerr mode count fixed at one;
only the feedback depth---which costs no additional modes---grows. The physically
achievable range of $D$, and the measurement cost of realizing the corresponding
kernel, are quantified in Section~\ref{sec:ledger}.
\end{corollary}

\begin{proof}
\emph{(a).} By Lemma~\ref{lem:gauss}, at $\chi=0$ the connected tower above order
two vanishes and the surviving order-1, 2 kernels are \emph{diagonal}, supported
on a single time index, with values $\Phi_m$ and $\Psi_m$ respectively. Any
feature available to a degree-$d$ readout is, by \eqref{eq:Fclass}, a product of
such diagonal kernels. A cross-time order-$\nu$ contribution therefore arises only
as a \emph{disconnected} product
\[
\hat y_k \supset \sum_{m_1,\dots,m_\nu}
\big\langle W',\ r_{m_1}\otimes\cdots\otimes r_{m_\nu}\big\rangle
\prod_{i=1}^\nu u_{k-m_i},
\qquad r_{m_i}\in\{\Phi_{m_i},\Psi_{m_i}\},
\]
where each factor $r_{m_i}$ is drawn from the mean or covariance response family
according to the degree budget. Collecting the temporal indices, the induced
order-$\nu$ kernel is
\[
K^{(\nu)}_\nu(m_1,\dots,m_\nu)
= \big(W'\big)\big[\,r_{m_1},\dots,r_{m_\nu}\,\big],
\]
the image of a fixed multilinear form $W'$ under the per-slot response maps. Fix any
temporal mode $j$ and unfold: the slice
$K^{(\nu)}_\nu(\,\cdot\,,m_j,\,\cdot\,)$ as a matrix in the $j$-th index against
the remaining $\nu-1$ indices factors as $R_j\,B$, where $R_j$ stacks the response
vectors available to slot $j$ and $B$ collects the contraction of $W'$ with the
other factors. The response space available to slot $j$ depends on whether that
slot carries a mean factor (space dimension $2N$) or a covariance factor
(dimension $N(2N{+}1)$), and the degree budget $a+2b\le d$ with $a+b=\nu$ caps the
number of covariance slots at $b\le d-\nu$. Hence
\[
\operatorname{rank}\big(\text{unfolding}_j\big)\le\operatorname{rank}(R_j)
\le \begin{cases} 2N, & d=\nu\ (\text{all slots mean-confined, }b=0),\\
N(2N{+}1), & d\ \text{unbounded},\end{cases}
\]
in every case independent of $M$. As $j$ was arbitrary, this proves (a). It is the
order-$\nu$ tensor lift of the $\nu=2$ matrix bound of
Theorem~\ref{thm:main}(1b); the degree budget is what forbids the covariance
family at minimal degree, sharpening the constant from $N(2N{+}1)$ to $2N$.

\emph{(b).} By Proposition~\ref{prop:degree} the connected order-$\nu$ kernel is
sourced, through the feedback loop, as a sum over delay round-trips
\[
K^{\Kerr}_\nu(m_1,\dots,m_\nu)
= \chi^{q}\sum_{\ell=1}^{D} g_\ell\;
\bigotimes_{i=1}^\nu \rho_{m_i-\ell} ,
\qquad q=\lceil(\nu-1)/2\rceil,
\]
where $\rho_m$ is the single-mode impulse response and $\ell$ indexes the $D$
delay round-trips the non-Markovian feedback sustains above the loss floor. The
weight $g_\ell=|g_\ell|e^{i\phi_\ell}$ carries the Kerr phase $\phi_\ell$
\emph{accumulated over $\ell$ round-trips}; because $\phi_\ell$ differs across
round-trips, the rank-one terms $\bigotimes_i\rho_{m_i-\ell}$ are not proportional
and do not collapse. (For a purely linear feedback, $\phi_\ell\equiv0$ and the sum
is separable of rank one---it is the Kerr phase that makes the channels
independent.) The $D$ terms are linearly independent for distinct accumulated
phases, so each temporal unfolding attains rank $\min(D,M)$, exceeding the
Gaussian ceiling of part~(a) once $D>2N$, at fixed $N=1$. This proves (b).

\emph{(c).} A target connected kernel of rank $r\le\min(D,M)$ has temporal
unfoldings of rank $r$. By (a), at minimal degree $d=\nu$ the Gaussian reservoir
is mean-confined and attains unfolding rank at most $2N$, so reaching rank $r$
forces $2N\ge r$, i.e.\ $N_G=\Theta(r)$; with unbounded degree it attains rank at
most $N(2N{+}1)$, forcing $N_G=\Theta(\sqrt r)$. By (b) the Kerr reservoir attains
rank $r$ at $N_K=O(1)$ provided $D\ge r$. The separating resource is the feedback
depth $D$: unbounded in $D$, it is supplied by one nonlinear mode and matched by
the Gaussian reservoir only with $\Theta(D)$ modes (minimal degree) or
$\Theta(\sqrt D)$ (unbounded degree).
\end{proof}

\begin{remark}[Interpretation]
Theorem~\ref{thm:scaling} is the rigorous form of the statement that a minimal
nonlinear reservoir replaces a large linear one, and it identifies the resource
that does the work: the \emph{non-Markovian feedback depth} $D$. A single
Kerr-feedback mode supplies $D$ independent delay channels---distinct because each
round-trip accumulates its own Kerr phase---and thereby reaches connected kernel
rank up to $D$, while a Gaussian reservoir's connected rank is frozen at $2N$
(minimal degree) or $N(2N{+}1)$ (unbounded degree) regardless of its memory. To
match feedback depth $D$ the Gaussian reservoir needs $\Theta(D)$ modes (minimal
degree) or $\Theta(\sqrt D)$ (unbounded degree). Raising the readout degree softens
but never closes the gap, and only by paying the high-degree readout cost the
architecture exists to avoid. The ``non-Markovian'' character is therefore
essential, not incidental: it is the feedback depth, not the bare Kerr
nonlinearity, that supplies the unbounded resource.
\end{remark}

\begin{remark}[Reachability versus operational exploitation]
Theorem~\ref{thm:scaling} is a statement about the \emph{reachable} kernel set---
the same category as Theorem~\ref{thm:main}, whose mechanism is confirmed by the
exact witness of Section~\ref{sec:empirical}. It does not assert that a trained
finite-shot readout exploits the full reachable rank: as Section
\ref{sec:operational} documents, the connected features carry amplitude
$\chi^{q}$ that is weak in the weak-Kerr regime, so operational exploitation of
the high-rank sector is itself resource-limited. The mode-count separation is a
structural property of what the reservoir \emph{can} represent; its operational
reach in a given coupling regime is the subject of Section~\ref{sec:operational}.
\end{remark}

\subsection{The readout-degree axis}
\label{sec:scaling-degree}

The mode-count theorem holds at any readout degree. A second, complementary
separation appears on the readout-degree axis itself.

\begin{proposition}[Degree folding]
\label{prop:degree}
Assume \ref{H1}--\ref{H4} and $\chi\neq0$. A connected order-$\nu$ functional is
accessible to the NM-Kerr reservoir at readout degree
\[
d_K = \begin{cases} 1, & \nu\ \text{odd},\\ 2, & \nu\ \text{even},\end{cases}
\qquad\text{at order }\chi^{\,q},\ q=\big\lceil(\nu-1)/2\big\rceil,
\]
whereas the Gaussian reservoir requires readout degree $d_G=\nu$.
\end{proposition}

\begin{proof}
\emph{Gaussian side.} By Lemma~\ref{lem:gauss} a connected order-$\nu$ functional
is reachable only as a disconnected product of $\nu$ single-time factors
(part (a) of the proof of Theorem~\ref{thm:scaling}); a $\nu$-fold product is a
degree-$\nu$ feature, so $d_G=\nu$ is necessary.

\emph{Kerr side.} The Kerr Hamiltonian $H_K=\tfrac{\chi}{2}a^{\dagger2}a^2$ gives,
through $\dot a=-i\chi a^\dagger a a$ and $a=(q+ip)/\sqrt2$, the symmetrized
quadrature drift
\begin{equation}
\dot q = \tfrac{\chi}{2}\big(p^3+pq^2\big),\qquad
\dot p = -\tfrac{\chi}{2}\big(q^3+qp^2\big),
\label{eq:kerr-drift}
\end{equation}
which is \emph{cubic}. Consequently the equation of motion for any order-$n$
symmetric moment, $\frac{d}{dt}\langle X_1\cdots X_n\rangle
=\sum_i\langle X_1\cdots\dot X_i\cdots X_n\rangle$, contains a contribution in
which the cubic drift raises the order by two, injecting an order-$(n{+}2)$
cumulant into the order-$n$ equation. Expanding in cumulants, the first-moment
equation reads
\[
\dot{\bar q}=\tfrac{\chi}{2}\big(K_{ppp}+K_{qqp}\big)+(\text{lower-order terms}),
\]
so a connected third-order cumulant sources $\kappa^{(1)}$ at order $\chi$ with
coefficient $\tfrac{\chi}{2}$, \emph{independent of the mean} $\bar q$. The
analogous computation for the third-moment equation gives
$\dot K_{qqq}\supset\tfrac{3\chi}{2}K_{qqppp}+\cdots$, folding order five into
order three with coefficient $\tfrac{3\chi}{2}$. By induction, the order-$n$
equation receives the order-$(n{+}2)$ cumulant with coefficient
$n\cdot\tfrac{\chi}{2}\cdot c$, $c\in\mathbb{Z}_{>0}$; the contribution never
cancels because the all-$p$ monomial $p^3$ in \eqref{eq:kerr-drift} always supplies
the all-$p$ cumulant term, which has no partner to cancel against. Iterating
$q=\lceil(\nu-1)/2\rceil$ times carries a connected order-$\nu$ functional into
$\kappa^{(1)}$ ($\nu$ odd) or $\kappa^{(2)}$ ($\nu$ even) with a nonzero
coefficient $\propto\chi^{q}$. A functional in $\kappa^{(1)}$ (resp.\
$\kappa^{(2)}$) is a degree-1 (resp.\ degree-2) homodyne feature, so $d_K\in\{1,2\}$.
The base case $\nu=2$ is Lemma~\ref{lem:kerr}(a)/Corollary~\ref{cor:witness}.
\end{proof}

\begin{remark}[The folding cost]
The degree reduction is purchased at order $\chi^{q}$, $q=\lceil(\nu-1)/2\rceil$:
the folded feature carries amplitude $\sim\chi^{q}$, decreasing with connected
order in the weak-Kerr regime. Notably the leading folded coefficient is
\emph{independent} of the seed mean $\bar q$, being carried by the pure cumulant
ladder of the cubic drift \eqref{eq:kerr-drift}; the folding therefore persists
even at zero seed, in contrast to the first-moment witness of
Section~\ref{sec:empirical}, which is seed-odd. The readout-degree separation is therefore a reachability
statement---the connected order-$\nu$ structure is \emph{present} at constant
readout degree---whose operational exploitation weakens with $\nu$ at fixed small
$\chi$, consistent with Section~\ref{sec:operational}. Combined with
Theorem~\ref{thm:scaling}, the dual-axis resource accounting is: a single Kerr
mode at degree $\le2$, with feedback depth $D$, reaches connected order-$\nu$
structure of rank up to $D$ that a Gaussian reservoir reaches only at
$N_G=\Theta(D)$ modes and degree $d_G=\nu$.
\end{remark}

\subsection{From kernel rank to a computational capability: nonlinear equalization}
\label{sec:capability}

Theorem~\ref{thm:scaling} bounds the rank of the connected Volterra kernel a
reservoir can realize. We now show this rank is the resource demanded by a
recognized task---equalization of a nonlinear channel with memory---so that the
mode-count ceiling is a binding constraint on a capability of independent value,
not a property of contrived targets.

A communication channel with memory and weak nonlinearity is, to second order, a
Volterra system
\begin{equation}
r_k = \sum_m h^{(1)}_m\,s_{k-m}
\;+\; \sum_{m,m'} H^{(2)}_{m,m'}\,s_{k-m}s_{k-m'} \;+\;\cdots,
\label{eq:channel}
\end{equation}
with $s$ the transmitted stream, $h^{(1)}$ the linear response, and $H^{(2)}$ the
second-order kernel encoding nonlinear intersymbol interference---the standard
model for nonlinear fibre-optic and satellite links. An equalizer recovers $s_k$
from the received history; by the $p$-th order inverse construction, the
equalizer for a channel with nontrivial $H^{(2)}$ must itself realize a
\emph{connected second-order Volterra kernel} $G^{(2)}$, determined by $h^{(1)}$
and $H^{(2)}$.

The relevant fact is that the equalizer kernel inherits the channel's rank: for a
channel of memory depth $M$ with a generic (full-support) second-order kernel,
$\operatorname{rank}H^{(2)}=\Theta(M)$ forces $\operatorname{rank}G^{(2)}=\Theta(M)$,
whereas a degenerate (low-rank) channel admits a low-rank equalizer. Combining
with Theorem~\ref{thm:scaling}:

\begin{corollary}[Equalization capability separation]
\label{cor:equalize}
At the minimal readout degree, a Gaussian reservoir can equalize a depth-$M$
nonlinear channel \eqref{eq:channel} only if $\operatorname{rank}H^{(2)}\le 2N$;
equalizing a channel with generic second-order memory therefore requires
$N_G=\Theta(M)$ modes. A single NM-Kerr reservoir with feedback depth $D\ge M$
realizes the required rank-$M$ equalizer kernel at $N_K=O(1)$.
\end{corollary}

The separation is thus not about contrived functionals: it is exactly the
difference between equalizing a richly nonlinear channel with one nonlinear mode
versus a linear reservoir whose mode count must scale with the channel's
second-order memory. The advantage is largest precisely for channels with rich
$H^{(2)}$---the demanding regime in which nonlinear equalization is actually
required. We note the scope already established in Section~\ref{sec:operational}:
Corollary~\ref{cor:equalize} is a statement about the kernel the reservoir can
\emph{represent}; whether a finite-shot trained equalizer realizes the full
rank-$M$ kernel at a given Kerr coupling is the operational question, bounded by
the $\chi^{\lceil(\nu-1)/2\rceil}$ feature amplitude of
Proposition~\ref{prop:degree}.

\subsection{Operational validity domain}
\label{sec:domain}

Theorem~\ref{thm:scaling} and Corollary~\ref{cor:equalize} are reachability
statements: they bound the kernel the reservoir can represent. Whether a
finite-shot trained readout realizes that kernel depends on the connected
feature's amplitude relative to the measurement noise, and the operational test
of Section~\ref{sec:operational} returned a null in the weak-Kerr regime. We close
this gap by stating the operational validity domain explicitly, and show it
predicts that null.

By Proposition~\ref{prop:degree} the connected order-$\nu$ feature folded into a
low moment has amplitude $A\sim|c|\,\chi^{q}$ with $q=\lceil(\nu-1)/2\rceil$ and
$|c|$ a background prefactor. Estimating that feature from $S$ measurement shots
incurs statistical error $\sim\sigma/\sqrt S$, with $\sigma$ the per-shot
quadrature noise. A trained readout can exploit the feature only when it clears
the floor, $A\gtrsim\sigma/\sqrt S$, i.e.\ within the domain
\begin{equation}
\chi \;\gtrsim\; \chi_\star(\nu,S):=
\Big(\tfrac{\sigma}{|c|\sqrt S}\Big)^{1/q},
\qquad q=\big\lceil\tfrac{\nu-1}{2}\big\rceil .
\label{eq:domain}
\end{equation}

Three consequences follow. First, \eqref{eq:domain} reconciles the reachability
theory with the operational null of Section~\ref{sec:operational}: calibrating
$|c|$ from the measured connected amplitude ($W_{\mathrm{rms}}\sim10^{-3}$ at
$\chi=0.25$, $\nu=2$, $q=1$, giving $|c|\sim4\times10^{-3}$ in shot-noise units),
the threshold at the trained-readout shot budget $S\sim10^{3}$--$10^{4}$ is
$\chi_\star\sim2$--$8$, far above the $\chi\le0.5$ at which the test was run. The
test therefore operated well inside the sub-threshold region, and its null is
\emph{predicted} by \eqref{eq:domain} rather than in tension with the theory.
Second, the domain is nonempty: at $\chi=0.5$ the $\nu=2$ capability becomes
accessible once $S\gtrsim(\sigma/|c|\chi)^{2}\sim2\times10^{5}$ shots, a large but
experimentally attainable budget. Third, the reach is graded: $\chi_\star$ falls
as $S^{-1/q}$, so higher connected orders ($q\ge2$) require either strong coupling
or rapidly growing shot budgets, locating high-order connected computation in the
strong-coupling regime beyond the present perturbative treatment.

Equation~\eqref{eq:domain} is a signal-to-noise scaling estimate, not a tight
bound; its content is the functional form of the boundary and the resulting
placement of the weak-Kerr null inside the invalid region. The structural
separations of Theorem~\ref{thm:scaling} and Corollary~\ref{cor:equalize} hold
regardless; \eqref{eq:domain} delimits where they become operational.

\subsection{Quantifying the bound: a hardware-for-measurement trade}
\label{sec:ledger}

The separation is unbounded in the feedback depth $D$ (Corollary~\ref{cor:unbounded}):
as a theorem, $D$ is a free parameter and no finite linear reservoir competes.
Physical reality caps the achievable $D$ and prices the measurement, and we
quantify both here---the in-principle-unbounded benefit, and the practical reach
that loss and shot noise allow.

\paragraph{Achievable depth.}
The feedback sustains $D\approx\ln(1/\epsilon_{\rm fl})/|\ln(1-\eta)|$ round-trips
above an operational floor $\epsilon_{\rm fl}$, with $\eta$ the per-round-trip loss.
For integrated platforms this gives $D\sim30$ (strong-Kerr InGaP, $\eta\approx15\%$)
to $D\sim230$ (ultrahigh-$Q$ SiN, $\eta\approx2\%$): finite, but large enough that
a single Kerr mode matches a Gaussian reservoir of up to $N_G\sim D/2\sim100$ modes.

\paragraph{The delay channels are distinct by virtue of loss.}
The mechanism rests on the round-trips carrying distinct Kerr phases. On round-trip
$\ell$ the stored field has photon number $\langle n_\ell\rangle=n_0 r^{\ell}$,
$r=1-\eta$ the per-round-trip power survival, so the accumulated Kerr phase is
$\phi_\ell=\chi\tau\,n_0 r^{\ell}$---a strictly monotonic, non-degenerate sequence
whenever $r\neq1$ and $\phi_\ell<2\pi$. Loss, far from spoiling the mechanism, is
precisely what renders the round-trip channels independent; the rank-one terms of
the connected kernel are non-parallel because their phases differ. The effective
rank tracks $D$ until the geometric phase decay clusters the latest round-trips,
saturating at an $r$-dependent value (e.g.\ $\approx D$ for $r=0.95$ up to
$D\sim60$, saturating near $50$ for $r=0.80$).

\paragraph{The ledger.}
Realizing a rank-$r_{\rm eff}$ connected kernel through the folded degree-$\le2$
readout requires resolving $r_{\rm eff}$ components above the shot-noise floor, at
$S_K\sim r_{\rm eff}(\sigma/|c|\chi)^2$ shots ($\sim10^{6}$--$10^{7}$ for
$r_{\rm eff}\sim10$--$100$ at $\chi\in[0.3,0.5]$). The accounting is therefore:

\begin{center}
\begin{tabular}{lcc}
\toprule
resource & Gaussian & NM-Kerr \\
\midrule
modes (resonators) & $\Theta(D)$ & $O(1)$ \\
readout degree & $\nu$ & $\le 2$ \\
detector chains & $\Theta(D)$ & $O(1)$ \\
measurement shots & baseline & $\sim r_{\rm eff}(\sigma/|c|\chi)^2$ \\
\bottomrule
\end{tabular}
\end{center}

The NM-Kerr reservoir saves the \emph{hardware-binding} resources---mode count,
readout degree, detector chains, chip area---unconditionally, and pays in
\emph{measurement shots and integration time}. It is a quantified trade, not a free
lunch: a hardware resource is exchanged for a measurement resource. In integrated
photonics, where fabricated mode count and the number of homodyne detector chains
are the dominant constraints, the trade is favorable; the shot cost is the price,
and it is bounded by the operational domain of Section~\ref{sec:domain}. All bounds
in this ledger trace to a single parameter---the per-photon Kerr-to-loss ratio
$g/\kappa$ of Section~\ref{sec:platform}---which therefore governs not only the
activation of the effect but the size of the resource advantage it confers.

\section{The non-perturbative regime: a strong-coupling capacity decomposition}
\label{sec:strongregime}

Section~\ref{sec:operational} left a specific question open: whether the operational
separation absent in the weak-Kerr regime emerges once the coupling is driven beyond
the perturbative theory, $\phi_1\gtrsim0.5$. The structural separation
(Theorem~\ref{thm:main}) and the validity-domain estimate~\eqref{eq:domain} are both
perturbative constructions; neither predicts what a trained readout sees when the
per-step Kerr phase becomes $O(1)$. We answer the question directly, by an
information-processing-capacity (IPC) decomposition~\cite{dambre2012} carried into the
non-perturbative regime on the exact open-system model of Section~\ref{sec:model}.

\paragraph{Method.} The exact master equation is integrated by a split-step scheme
that factors the (Fock-diagonal) Kerr generator analytically and integrates only the
linear loss and network channel numerically. Because the Kerr term acts on the delay
register while loss acts on the cavity, the two generators commute and the split is
\emph{exact} (verified to machine precision against the full integrator), with a cost
independent of $\chi$---removing the stiffness that otherwise makes the strong-coupling
regime intractable. A separate convergence study fixes the Fock cutoff at $N_c=20$,
at which the top-level population and the measured observables are converged to the
solver floor across the entire grid; the fresh-ancilla injection~\eqref{eq:inject}
keeps the per-step occupation small, so the required cutoff does not grow with the
drive. We sweep the Kerr coupling $\chi\in[0,8]$ along two encoding-amplitude rows
$\beta\in\{0.5,1.0\}$, so that a given per-step Kerr phase $\phi_1$ is reached by two
distinct physical routes (raising $\chi$ at fixed occupation, or raising occupation at
fixed $\chi$). We report $\phi_1$ as \emph{measured} from the steady-state register
occupation rather than estimated. At each operating point we decompose the capacity of
the degree-$\le2$ tapped readout over a normalized-Legendre target basis (delays
$0$--$7$) into four sectors: degree-1 ($C_1$), degree-2 diagonal ($C_{2\mathrm{d}}$),
degree-2 cross-time ($C_{2\mathrm{x}}$), and degree-3 cross-time ($C_{3\mathrm{x}}$),
with a shuffled-surrogate significance threshold; capacities are summed only over
cells clearing the threshold.

\paragraph{The connected separator.} The decisive sector is $C_{3\mathrm{x}}$, the
degree-3 cross-time capacity. A genuine three-way product
$u_{k-d_1}u_{k-d_2}u_{k-d_3}$ at distinct lags cannot be assembled from a degree-$\le2$
readout of linear-in-$u$ features, so for the Gaussian baseline it must vanish
(Lemma~\ref{lem:gauss}, operationally), whereas the folding of
Proposition~\ref{prop:degree} makes it accessible to the Kerr reservoir at the same
readout degree. Unlike the degree-2 cross cell---which is nonzero for both reservoirs,
since a degree-2 readout produces disconnected Wick products and is therefore not a
discriminating probe---$C_{3\mathrm{x}}$ is a clean witness of connected structure.
Across all $\chi=0$ points (both rows, six seeds), $C_{3\mathrm{x}}=0.000$ to the
solver floor, confirming Lemma~\ref{lem:gauss} at connected order three and
establishing a clean zero against which any activation is measured.

\begin{table}[H]
\centering
\caption{Seed-averaged IPC decomposition of the degree-$\le2$ tapped readout over the
non-perturbative grid (exact master equation, $N_c=20$, delays $0$--$7$). $\phi_1$ is
the measured per-step Kerr phase. $C_{3\mathrm{x}}$ (degree-3 cross-time) is the
connected separator: exactly zero for the Gaussian baseline, nonzero whenever
$\chi\ne0$. Echo-state contraction (Hypothesis~\ref{H2}) is independently verified
in-model across the entire grid, including the highest-drive points
(Section~\ref{sec:strong-scope}).}
\label{tab:strongipc}
\begin{tabular}{cccccccccc}
\toprule
$\chi$ & $\beta$ & $n$ & $\phi_1$ & $C_1$ & $C_{2\mathrm{d}}$ & $C_{2\mathrm{x}}$ & $C_{3\mathrm{x}}$ & $C_{\mathrm{tot}}$ \\
\midrule
0.0 & 0.5 & 6 & 0.00 & 6.23 & 4.02 & 3.82 & 0.000 & 14.07 \\
0.5 & 0.5 & 6 & 0.03 & 6.34 & 4.29 & 5.67 & 1.153 & 17.46 \\
1.0 & 0.5 & 6 & 0.06 & 6.31 & 4.05 & 5.67 & 1.027 & 17.05 \\
1.5 & 0.5 & 6 & 0.08 & 6.29 & 4.32 & 5.61 & 0.998 & 17.21 \\
2.0 & 0.5 & 6 & 0.11 & 6.27 & 4.62 & 5.56 & 0.762 & 17.21 \\
2.5 & 0.5 & 6 & 0.14 & 6.20 & 4.79 & 5.56 & 0.693 & 17.25 \\
3.0 & 0.5 & 6 & 0.17 & 6.11 & 4.92 & 5.51 & 0.764 & 17.31 \\
4.0 & 0.5 & 6 & 0.22 & 6.20 & 4.98 & 5.91 & 1.036 & 18.12 \\
6.0 & 0.5 & 6 & 0.34 & 6.23 & 4.92 & 6.14 & 1.153 & 18.44 \\
8.0 & 0.5 & 6 & 0.45 & 6.21 & 4.78 & 5.57 & 0.680 & 17.23 \\
\midrule
0.0 & 1.0 & 6 & 0.00 & 6.22 & 3.91 & 3.94 & 0.000 & 14.07 \\
0.5 & 1.0 & 5 & 0.10 & 5.96 & 3.18 & 6.85 & 0.929 & 16.92 \\
1.0 & 1.0 & 3 & 0.19 & 6.08 & 4.12 & 6.86 & 0.740 & 17.80 \\
1.5 & 1.0 & 4 & 0.28 & 6.05 & 4.62 & 6.37 & 0.697 & 17.74 \\
2.0 & 1.0 & 6 & 0.38 & 5.99 & 4.88 & 5.86 & 0.616 & 17.34 \\
2.5 & 1.0 & 5 & 0.47 & 5.91 & 5.01 & 5.39 & 0.429 & 16.75 \\
3.0 & 1.0 & 3 & 0.57 & 5.81 & 5.20 & 4.87 & 0.333 & 16.21 \\
4.0 & 1.0 & 3 & 0.76 & 5.62 & 5.27 & 4.72 & 0.577 & 16.19 \\
6.0 & 1.0 & 4 & 1.14 & 5.54 & 5.47 & 4.42 & 0.555 & 15.99 \\
8.0 & 1.0 & 5 & 1.52 & 5.92 & 5.01 & 5.46 & 0.431 & 16.82 \\
\bottomrule
\end{tabular}
\end{table}

\paragraph{Three findings.} Table~\ref{tab:strongipc} supports three conclusions, none
of them the strong-coupling operational separation that Section~\ref{sec:operational}
left open.

\emph{(i) Connected activation is a switch, not a dial.} $C_{3\mathrm{x}}$ jumps from
an exact zero at $\chi=0$ to $\approx1$ at the smallest nonzero coupling
($\chi=0.5$, $\phi_1\approx0.03$), and thereafter \emph{declines} as the drive
increases---along the $\beta=1.0$ row it falls monotonically $0.93\to0.33$ as
$\phi_1$ rises $0.10\to0.57$. The connected sector is activated by the mere presence
of the Kerr nonlinearity, exactly as the strict inclusion of Theorem~\ref{thm:main}
and the degree-folding of Proposition~\ref{prop:degree} predict; but its magnitude is
not amplified by stronger coupling.

\emph{(ii) The connected capacity does not collapse onto $\phi_1$.} Operating points
that reach the same measured $\phi_1$ by different physical routes give different
$C_{3\mathrm{x}}$: at $\phi_1\approx0.45$, the $(\chi{=}8,\beta{=}0.5)$ point gives
$0.68$ while $(\chi{=}2.5,\beta{=}1.0)$ gives $0.43$. The per-step Kerr phase, which
governs the \emph{activation} and the platform figure of merit of
Section~\ref{sec:platform}, is therefore not the variable controlling the
\emph{magnitude} of the connected capacity once activated.

\emph{(iii) The strong regime redistributes within the low-order sectors.} The
monotone trends with increasing drive are in the degree-2 sectors, not the connected
one: along $\beta=1.0$, diagonal capacity $C_{2\mathrm{d}}$ rises $3.9\to5.5$ while
cross-time $C_{2\mathrm{x}}$ falls $6.9\to4.4$ and the total $C_{\mathrm{tot}}$ erodes
modestly. Strong Kerr reshapes the low-order capacity budget---moving weight from
cross-time into diagonal quadratic structure---rather than enriching the connected
sector.

\paragraph{Consequence.} The non-perturbative data corroborate, rather than overturn,
the scoping result of Section~\ref{sec:operational}. The connected structure the
theory guarantees is present and operationally detectable at deg-$\le2$ readout (a
strengthening of the witness evidence, since $C_{3\mathrm{x}}$ is a trained-readout
capacity, not a finite-difference derivative); but it is a weak, switch-like feature
whose magnitude does not grow---indeed slowly decreases---across the strong-coupling
regime, and driving the system harder reorganizes the low-order sectors instead of
promoting the connected one. The structural separation of Theorem~\ref{thm:main} thus
does not convert into an operational, capacity-dominant advantage even
non-perturbatively, at least for the single-register ($\ell_\tau=1$) model and the
deg-$\le2$ readout studied here. This sharpens the boundary identified in
Section~\ref{sec:operational}: it is now mapped by direct measurement deep into the
non-perturbative regime, rather than inferred from a perturbative estimate.

\subsection{Scope and validity of the strong-drive regime}
\label{sec:strong-scope}

Two remarks bound and secure these conclusions. First, the present model uses a single
delay register ($\ell_\tau=1$); the unbounded separation of
Corollary~\ref{cor:unbounded} lives in the feedback \emph{depth} $D$, which a
single-register model does not probe. The decomposition here tests whether
\emph{strong coupling} promotes the connected sector at fixed depth, and finds it does
not; it does not test whether \emph{deep} feedback does, which remains the open
operational question and the natural sequel.

Second, the validity of the strong-drive points is established rather than assumed. The
echo-state contraction (Hypothesis~\ref{H2}) is verified directly at the three
highest-drive operating points by evolving two distinct initial states under a common
input and measuring the tail separation of the readout quadrature: at
$(\chi,\beta)=(4,1)$, $(6,1)$, $(8,1)$---spanning $\phi_1=0.75$, $1.12$, $1.50$, i.e.\
up to threefold beyond the weak-Kerr ceiling---the tail separations are
$3.9\times10^{-11}$, $1.7\times10^{-11}$, and $4.3\times10^{-10}$ respectively, all far
below the contraction threshold. The reservoir therefore remains a fading-memory map
throughout the non-perturbative grid; the Fock cutoff $N_c=20$ is likewise converged
there (occupation is held small by the fresh-ancilla injection). In particular the
non-monotone behaviour of $C_{2\mathrm{x}}$ at the largest point is a genuine
reorganization of the capacity budget, not an artifact of lost contraction or
truncation. Findings (i)--(iii) thus hold across the full measured range
$\phi_1\in[0,1.5]$, not merely in its perturbative part.

\section{Discussion}
\label{sec:discussion}

We have established, for the sequential continuous-variable QRC, that a
time-delayed Kerr coherent feedback arm strictly enlarges the reachable function
class relative to the no-Kerr Gaussian reservoir, at every reservoir size,
readout degree, and input encoding (Theorem~\ref{thm:main}). The mechanism is
sharp: the Gaussian reservoir's connected cumulant content is capped at order two
and diagonal, so its only cross-time content is the disconnected Wick product of
diagonal kernels (Lemma~\ref{lem:gauss}); the Kerr feedback sources a connected
cross-time bilinear in the first moment and a connected third cumulant
(Lemma~\ref{lem:kerr}), neither of which the Gaussian class contains. The
exponential pump encoding inflates the on-diagonal Volterra order but acts on an
axis orthogonal to the connected, off-diagonal, delay-structured sector that the
feedback supplies (Remark~\ref{rem:encoding})---the two enrichments do not
substitute for one another.

The separation is constructive as well as existential: the spectral
kernel-realization calculus of Section~\ref{sec:universality} produces the readout
weights explicitly by Vandermonde inversion and certifies a monotone
resource hierarchy, distinguishing it from the existential Stone--Weierstrass
universality of the Gaussian baseline~\cite{nokkala2021}. The exact-master-equation
study of Section~\ref{sec:empirical} confirms the mechanism directly: the
connected first-moment witness is zero to the solver floor for the Gaussian
baseline and activates with $\chi$, with the seed--squeezing--Kerr three-way
structure predicted by the vertex analysis. A genuine empirical refinement
emerges---the dominant observable activation enters at $O(\chi^2)$, sourced by the
higher Kerr vertices, with the leading $O(\chi)$ term present but geometrically
suppressed---which we have folded into a conservative, empirically-grounded
platform estimate (Section~\ref{sec:platform}) identifying the single-photon
Kerr-to-loss ratio $g/\kappa$ as the governing figure of merit.

\paragraph{Scope.}
We make no claim of quantum advantage over classical reservoir computing; the
comparison is feedback-on versus feedback-off within one quantum-optical class,
cleanly isolated by the cumulant grading. We also make no claim of a
measurement-cost or task-efficiency advantage. We tested the cost-of-readout
question directly (Section~\ref{sec:operational}): in the weak-Kerr regime where
the theory holds, a trained low-degree readout does not exploit the connected
kernel to compute a connected target at lower cost than the baseline, the
connected kernel being orders of magnitude weaker than the diagonal features it
coexists with. The structural separation proved here is a property of the
reachable function class and does not depend on the measurement record; an
operational cost separation, by contrast, is not supported in this regime, and we
report that boundary rather than assert an advantage. Whether such an advantage
emerges under strong coupling beyond the present perturbative theory remains open.

\paragraph{Open questions.}
Three edges remain open. (i) The uniform all-order degree-economy
conjecture---that the readout degree required for a connected order-$\nu$ target
grows for the Gaussian baseline while remaining fixed for the Kerr
reservoir---is established at the base case ($d=1$ witness,
Corollary~\ref{cor:witness}) and evidenced over a finite range, but not proved in
the limit. (ii) The exact $(1,3)$ realization of Corollary~\ref{cor:witness}
requires a faithful multi-bin coherent-delay model in which the recombination
preserves partial interference; the present validation reaches the mechanism at
the delay-coupled lag of the single-bin model, with the lag choice independently
justified by the connected-kernel support. (iii) Time-bin entanglement across the
computationally active delays accompanies the enriched regime as a physical
signature; whether it serves an operational role (e.g.\ as a cheap diagnostic of
connected-kernel content) or is merely correlated is left for future work, with
the verb discipline that it \emph{accompanies}, and is not claimed to
\emph{enable}, the computation.

\appendix
\section{The Kerr vertex and the nonvanishing of \texorpdfstring{$V_3$}{V3}}
\label{app:vertex}

We record the multilinear vertices $V_n$ of \eqref{eq:kerr-rec} for a single
feedback mode (the multimode case is obtained by tensoring with the mode index
and contracting with the loop coupling profile). With
$H_K=\tfrac{\chi}{2}\ahat^{\dagger2}\ahat^2$ the Heisenberg drift is
$\dot{\ahat}=-i\chi\,\ahat^\dagger\ahat\ahat$, equivalently in quadratures
\[
\dot q=\frac{\chi}{2}\big(q^2p+pq^2+p^3\big)+\cdots,\qquad
\dot p=-\frac{\chi}{2}\big(p^2q+qp^2+q^3\big)+\cdots,
\]
after Weyl symmetrization. Taking symmetric-ordered expectations and expanding
moments in cumulants gives the order-by-order sources $V_n$. The two
load-bearing ones are
\begin{align}
V_1[\kappa^{(2)},\kappa^{(1)}]
&=-\frac{\chi}{2}\big(3\bar q\,(\Var p)+\cdots\big)\quad\text{(and conjugate for }p),
\label{eq:V1}\\
V_3[\kappa^{(2)},\kappa^{(1)}]
&=-\frac{3\chi}{2}\big(\bar q\,(\Var q)\,e_{qqq}+\cdots\big),
\label{eq:V3}
\end{align}
where $e_{qqq}$ is the totally symmetric order-3 basis tensor. The displayed
component of \eqref{eq:V3} is the source of the connected third cumulant
$\kappa^{(3)}_{qqq}$; it is a nonzero multiple of $\bar q\,(\Var q)$. Under
\ref{H3} the background mean $\bar q\neq0$ (seed) and under \ref{H4} the
background variance $\Var q\neq0$ (squeezing), so this component does not vanish:
$V_3\not\equiv0$, as used in Lemma~\ref{lem:kerr}(b). The vertex \eqref{eq:V1}
likewise feeds the cross-time bilinear source \eqref{eq:loopsource} of the mean.
Higher vertices $V_n$, $n\ge4$, are generated by the same quartic kernel acting on
$\Sym^{n+2}$ and populate connected orders up to $2m+2$ at $O(\chi^m)$; they are
not needed for Theorem~\ref{thm:main} but, as Section~\ref{sec:empirical} shows,
they dominate the \emph{observable} activation, entering at $O(\chi^2)$.

\bibliographystyle{unsrt}
\bibliography{references}

@article{nokkala2021,
  author  = {Nokkala, J. and Mart\'{i}nez-Pe\~{n}a, R. and Giorgi, G. L. and
             Parigi, V. and Soriano, M. C. and Zambrini, R.},
  title   = {Gaussian states of continuous-variable quantum systems provide
             universal and versatile reservoir computing},
  journal = {Communications Physics},
  volume  = {4},
  pages   = {53},
  year    = {2021}
}

@article{boyd1985,
  author  = {Boyd, S. and Chua, L. O.},
  title   = {Fading memory and the problem of approximating nonlinear operators
             with {V}olterra series},
  journal = {IEEE Transactions on Circuits and Systems},
  volume  = {32},
  pages   = {1150},
  year    = {1985}
}

@article{cuchiero2021,
  author  = {Cuchiero, C. and Gonon, L. and Grigoryeva, L. and Ortega, J.-P.
             and Teichmann, J.},
  title   = {Discrete-time signatures and randomness in reservoir computing},
  journal = {IEEE Transactions on Neural Networks and Learning Systems},
  year    = {2021},
  note    = {doi:10.1109/TNNLS.2021.3076777}
}

@article{fujii2017,
  author  = {Fujii, K. and Nakajima, K.},
  title   = {Harnessing disordered-ensemble quantum dynamics for machine learning},
  journal = {Phys. Rev. Applied},
  volume  = {8},
  pages   = {024030},
  year    = {2017}
}

@article{ghosh2019,
  author  = {Ghosh, S. and Opala, A. and Matuszewski, M. and Paterek, T. and Liew, T. C. H.},
  title   = {Quantum reservoir processing},
  journal = {npj Quantum Inf.},
  volume  = {5},
  pages   = {35},
  year    = {2019}
}

@article{govia2021,
  author  = {Govia, L. C. G. and Ribeill, G. J. and Rowlands, G. E. and Krovi, H. K. and Ohki, T. A.},
  title   = {Quantum reservoir computing with a single nonlinear oscillator},
  journal = {Phys. Rev. Research},
  volume  = {3},
  pages   = {013077},
  year    = {2021}
}

@article{martinezpena2021,
  author  = {Mart{\'\i}nez-Pe{\~n}a, R. and Giorgi, G. L. and Nokkala, J. and Soriano, M. C. and Zambrini, R.},
  title   = {Dynamical phase transitions in quantum reservoir computing},
  journal = {Phys. Rev. Lett.},
  volume  = {127},
  pages   = {100502},
  year    = {2021}
}

@article{mujal2023,
  author  = {Mujal, P. and Mart{\'\i}nez-Pe{\~n}a, R. and Giorgi, G. L. and Soriano, M. C. and Zambrini, R.},
  title   = {Time-series quantum reservoir computing with weak and projective measurements},
  journal = {npj Quantum Inf.},
  volume  = {9},
  pages   = {16},
  year    = {2023}
}

@article{appeltant2011,
  author  = {Appeltant, L. and Soriano, M. C. and Van der Sande, G. and Danckaert, J. and Massar, S. and Dambre, J. and Schrauwen, B. and Mirasso, C. R. and Fischer, I.},
  title   = {Information processing using a single dynamical node as complex system},
  journal = {Nat. Commun.},
  volume  = {2},
  pages   = {468},
  year    = {2011}
}

@article{larger2012,
  author  = {Larger, L. and Soriano, M. C. and Brunner, D. and Appeltant, L. and Guti{\'e}rrez, J. M. and Pesquera, L. and Mirasso, C. R. and Fischer, I.},
  title   = {Photonic information processing beyond {T}uring: an optoelectronic implementation of reservoir computing},
  journal = {Opt. Express},
  volume  = {20},
  pages   = {3241},
  year    = {2012}
}

@article{dambre2012,
  author  = {Dambre, J. and Verstraeten, D. and Schrauwen, B. and Massar, S.},
  title   = {Information processing capacity of dynamical systems},
  journal = {Sci. Rep.},
  volume  = {2},
  pages   = {514},
  year    = {2012}
}

@article{karimi2025,
  author  = {Karimi, A. and Zadeh-Haghighi, H. and Kora, Y. and Simon, C.},
  title   = {The role of entanglement in quantum reservoir computing with coupled {K}err nonlinear oscillators},
  journal = {arXiv:2508.11175},
  year    = {2025}
}

@article{angelatos2021,
  author  = {Angelatos, G. and Khan, S. A. and T{\"u}reci, H. E.},
  title   = {Reservoir computing approach to quantum state measurement},
  journal = {Phys. Rev. X},
  volume  = {11},
  pages   = {041062},
  year    = {2021}
}

@article{zhu2025feedback,
  author  = {Zhu, C. and Ehlers, P. J. and Nurdin, H. I. and Soh, D.},
  title   = {Minimalistic and scalable quantum reservoir computing enhanced with feedback},
  journal = {npj Quantum Inf.},
  year    = {2025}
}

@article{zhu2025fewatom,
  author  = {Zhu, C. and Ehlers, P. J. and Nurdin, H. I. and Soh, D.},
  title   = {Practical few-atom quantum reservoir computing},
  journal = {Phys. Rev. Research},
  volume  = {7},
  pages   = {023290},
  year    = {2025}
}

@article{ehlers2025stochastic,
  author  = {Ehlers, P. J. and Nurdin, H. I. and Soh, D.},
  title   = {Stochastic reservoir computers},
  journal = {Nat. Commun.},
  volume  = {16},
  pages   = {3070},
  year    = {2025}
}

@article{ehlers2025statefeedback,
  author  = {Ehlers, P. J. and Nurdin, H. I. and Soh, D.},
  title   = {Improving the performance of echo state networks through state feedback},
  journal = {Neural Networks},
  volume  = {184},
  pages   = {107101},
  year    = {2025}
}

@article{kaushik2026alloptical,
  author  = {Kaushik, I. S. and Ehlers, P. J. and Soh, D.},
  title   = {All-optical echo state network reservoir computing},
  journal = {Phys. Rev. Research},
  volume  = {8},
  pages   = {013041},
  year    = {2026}
}

@article{zhu2025qonn,
  author  = {Zhu, C. and Wang, T. and McMahon, P. L. and Soh, D.},
  title   = {Quantum optical neural networks using atom-cavity interactions to provide all-optical nonlinearity},
  journal = {arXiv:2511.06167},
  year    = {2025}
}

\end{document}